\documentclass[%
 reprint,
superscriptaddress,
 amsmath,
 amssymb,
 aps,
]{revtex4-2}

\usepackage[english]{babel}
\usepackage{graphicx}
\usepackage{dcolumn}
\usepackage{bm}
\usepackage{subcaption}
\usepackage[normalem]{ulem}
\usepackage{xcolor}
\usepackage{adjustbox}
\usepackage{braket}

\newcommand{\Zplus}{$\ket 0$~}
\newcommand{\Zminus}{$\ket 1$~}
\newcommand{\Xplus}{$\ket +$~}
\newcommand{\Yplus}{$\ket i$~}

\newcommand{\kket}[1]{{|#1 \rangle\rangle}}
\newcommand{\scC}{\mathcal{C}}
\usepackage{dsfont}
\newcommand{\figref}[1]{Fig.\,\ref{#1}}

\usepackage{amssymb,amsmath,amsthm}
\newtheorem{theorem}{Theorem}
\newtheorem{corollary}{Corollary}[theorem]

\begin{document}

\preprint{APS/123-QED}

\title{Syncopated Dynamical Decoupling for Suppressing Crosstalk in Quantum Circuits}

\author{Bram Evert}
  \email{bevert@rigetti.com}
  \affiliation{Rigetti Computing, Berkeley, CA}
\author{Zoe Gonzalez Izquierdo}
  \affiliation{Quantum Artificial Intelligence Laboratory (QuAIL), NASA Ames Research Center, Moffett Field, CA}
  \affiliation{Research Institute for Advanced Computer Science (RIACS), USRA, Moffett Field, CA}
\author{James Sud}
  \affiliation{University of Chicago Computer Science Department, IL}
  \affiliation{Research Institute for Advanced Computer Science (RIACS), USRA, Moffett Field, CA}
\author{Hong-Ye Hu}
  \altaffiliation{Presently at Harvard University.}
  \affiliation{Quantum Artificial Intelligence Laboratory (QuAIL), NASA Ames Research Center, Moffett Field, CA}
  \affiliation{Research Institute for Advanced Computer Science (RIACS), USRA, Moffett Field, CA}
\author{Shon Grabbe}
  \affiliation{Quantum Artificial Intelligence Laboratory (QuAIL), NASA Ames Research Center, Moffett Field, CA}
\author{Eleanor G. Rieffel}
  \affiliation{Quantum Artificial Intelligence Laboratory (QuAIL), NASA Ames Research Center, Moffett Field, CA}
\author{Matthew J. Reagor}
  \altaffiliation{Presently at Google Quantum AI}
  \affiliation{Rigetti Computing, Berkeley, CA}
\author{Zhihui Wang}
  \altaffiliation{zhihui.wang@nasa.gov}
  \affiliation{Quantum Artificial Intelligence Laboratory (QuAIL), NASA Ames Research Center, Moffett Field, CA}
  \affiliation{Research Institute for Advanced Computer Science (RIACS), USRA, Moffett Field, CA}

\date{\today}

\begin{abstract}
Theoretically understanding and experimentally characterizing and modifying the underlying Hamiltonian of a quantum system is of utmost importance in achieving high-fidelity quantum gates for quantum computing. In this work, we explore the use of dynamical decoupling (DD) in characterizing and suppressing undesired two-qubit couplings as well as the underlying single-qubit decoherence, both significant hurdles to achieving precise quantum control and realizing quantum computing on many hardware prototypes.  Through discrete search of dynamical decoupling sequences, we identify sequences that protect against decoherence and selectively target unwanted two-qubit interactions of general form.  On a transmon-qubit-based superconducting quantum device, we identify separate white and 1/f noise components underlying the single-qubit decoherence and a static ZZ coupling between pairs of qubits. A family of syncopated dynamical decoupling sequences is found and their efficiency demonstrated in two-qubit benchmarking experiments.  The syncopated decoupling technique significantly boosts performance in a realistic algorithmic quantum circuit.
\end{abstract}

\maketitle

\section{\label{sec:level1}Introduction}

Error suppression techniques play a crucial role in exploring applications of current noisy quantum hardware as well as in achieving error threshold for fault-tolerant quantum computing~\cite{Gottesman2009}. Among them, dynamical decoupling (DD) has emerged as a powerful strategy~\cite{Viola1998, Viola1999, Viola1999_2, Zanardi1999, Vitali1999, Wu2002, Khodjasteh2005, Biercuk2009, Souza2012}. By applying a sequence of pulses to a qubit, DD engineers an effective Hamiltonian that averages out unwanted inter-qubit couplings or the qubit interaction with the environment~\cite{Waugh1968_2}.  With its simplicity in concept, design and implementation, and great versatility in application, DD has a long track record of successfully suppressing single-qubit decoherence on different quantum device types and noise of various power spectra~\cite{Alvarez_PRL_2011,Cywinski_2017JPCM}. 

Contemporary quantum computers often suffer from unwanted interactions between qubits, deemed ``crosstalk''~\cite{Sarovar2020}. Such crosstalks are a leading error term in two-qubit gates and on idling qubits in circuit contexts. While DD can effectively suppress single-qubit decoherence, if identical periodic DD sequences are applied to coupled qubits in a {\it synchronized} manner, the crosstalk among them will not be mitigated. It has long been known that for two qubits coupled in $ZZ$ form, if a periodic decoupling sequence is applied to one of them while the other is left idle, the $ZZ$ effect on the idle qubit will be averaged out~\cite{Garbow1982, Pokharel2018}. Ref.~\cite{Tripathi2022} demonstrated this experimentally, showing that a single-qubit DD could be used to decouple crosstalk between a data qubit and idle qubits coupled to it. 

Recently, using DD to suppress ZZ crosstalk in more complex scenarios has been a topic of great interest. Ref. \cite{Zhou_PRL_2023}, demonstrated that the timing of XY4 sequences can be shifted to suppress ZZ crosstalk between pairs of qubits, while simultaneously improving coherence, and ref. \cite{Shirizly2024} explores the use of shifted XX sequences for the same effect. Refs. \cite{Mundada2023,Carrera_Vazquez_2024}, describe a more general approach of using ``staggered DD'' to suppress ZZ crosstalk while keeping the number of pulses as low as possible. 
Here, we present a general framework that incorporates these techniques and describe it as forms of ``syncopation'', which combines the use of time-shifting, frequency-doubling and operator-alternation to suppress a known network of static crosstalks. Following the spirit of the combinatorial approach for DD \cite{LidarQEC}, we translate finding the optimal sequence into a discrete optimization problem over the superoperator representation. We find that this technique can generate sequences for suppressing crosstalk in the form of a broad family of two-body couplings. The term ``syncopation'' is borrowed from music, referring to the practice of playing a rhythm off-beat, which is a useful analogy to understand the mechanism of suppression. A family of syncopated DD also retains the mechanism of suppressing dephasing on individual qubits, therefore protecting all the qubits in the system and serving as a powerful approach for crosstalk suppression in quantum systems. We also show how the concept of syncopation can be applied to crosstalk characterization by suppressing subsets of crosstalks. Finally, we combine dynamical decoupling with randomized compiling to improve the performance of an application circuit in a scalable way. 

Our experimental platform is a Rigetti Aspen chip, with fixed couplings between the qubits. This architecture suffers from a common crosstalk found on superconducting platforms of the form $e^{i\theta \text{ZZ}}$, which provides a useful test case for applying the technique. While the superconducting qubit architecture used here contains only ZZ crosstalk, other architectures demonstrate a variety of cases {\cite{sundaresan_reducing_2020, sete_floating_2021, ospelkaus_trapped-ion_2008, fang_crosstalk_2022}}. 

We first measure the magnitude of the ZZ crosstalk, and the effectiveness of syncopated dynamical decoupling in suppressing it, using Ramsey measurements on pairs of qubits. The Ramsey experiments showcase a beating pattern in the observables with coupled qubits, a signature of the $\text{ZZ}$ influence on initial state $\ket{++}$. We further demonstrate that the beating signal can be used to obtain an accurate measurement of the ZZ coupling, consistent with alternative measurements, and suggest how syncopation can be applied for large-scale, efficient measurements of crosstalks and bare qubit frequencies. By applying syncopated DD sequences, we show that the ZZ coupling and associated beating can be suppressed. Further applying DD to idle qubits in a quantum circuit for Quantum Alternating Operator Ansatz (QAOA) led to a significant boost in the circuit performance. We also demonstrate how syncopated DD can work in conjunction with randomized compilation, another error mitigation method, to provide even greater joint benefit. 

We also exploited the fact that synchronized DD, while not suppressing the crosstalk, serves as an approach to single out this crosstalk effect from single-qubit noise and hence serves as an alternative way to characterize it.  We demonstrate a precise measurement of the magnitude of the static ZZ coupling, showing high consistency with lower-level measurement based on hardware setup.  Such precise measurement opens a path to pin the bare frequency of a qubit, complementing the information provided in standard calibration where effect of surrounding qubits often factors in the reported qubit frequency.

During the preparation of this manuscript, we became aware of similar work. Ref.~\cite{Niu2024} presents a set of experiments which showcase the effectiveness of syncopated DD in protecting pairs of qubits. This is similar to our results in Section \ref{subsec:2_qubit_zz}, but explores a wider variety of sequences. Ref.~\cite{Seif2024} introduces an compilation scheme called Context-Aware Dynamical Decoupling (CA-DD) which is similar to our syncopated compilation scheme described in Sections \ref{subsec:application} and \ref{sec:syncopation-on-a-graph}. Finally, after these simultaneous papers, \cite{coote_resource-efficient_2024} described a more general method of finding multi-qubit sequences called GraphDD, which considers the common situation in which qubits are not idle for identical periods of time.

\section{\label{sec:Background}Sequence design for dynamical decoupling}
The theory of DD originated in the field of nuclear magnetic resonance (NMR)~\cite{Hahn1950, Carr1954, Meiboom1958, Waugh1968, Waugh1968_2, Waugh1982, Garbow1982, Vandersypen2005}, as a way to enhance the coherence time of a collection of nuclear spins, by utilizing fast control pulse sequences to average out the effects of noise. Its development has been quickly gaining momentum thanks to the advances in quantum technology; the ability to precisely control individual qubits has allowed the use of DD to prolong their coherence time. A variety of DD sequences have been designed~\cite{Maudsley1986, Guillon1990, Biercuk2011,Cywinski_2017JPCM} to target the noise spectra that arise for different physical realizations of qubits.

Our primary target in this study is a static coupling between qubits, such that the crosstalk can be modeled by a constant Hamiltonian $H$. We thus will focus on the operator aspect of the DD design, i.e., identifying a set of single-qubit operators (pulses) that has the potential of achieving a decoupled effective Hamiltonian.
To generalize into the non-static case, the positioning of these pulses could be adjusted from studying the temporal features of the crosstalk.


In this section we first detail the method of designing DD sequences for an arbitrary $H$ through discrete optimization in the superoperator representation, then introduce the concept of syncopation, and finally focus on the ZZ crosstalk for superconducting quantum devices.

\subsection{DD design in superoperator representation}
\label{sec:superoperator} 

The Pauli operators comprise an orthogonal basis in the Hilbert space, and any Hermitian operator $\text{A}\in \mathds{C}^{2\times 2}$ can be represented as a vector, $\text{A}=\alpha I+\beta X +\gamma Y +\zeta Z$ or $\kket{A}=(\alpha,\beta,\gamma,\zeta)$, where $\alpha,\beta,\gamma,\zeta \in \mathbb{R}$ \footnote{Note here the basis is not normalized. It can be normalized by assigning a prefactor of $1/\sqrt{D}$ where $D$ is the dimension of the Hilbert space.}.
For a two-qubit system, the basis is composed of 16 Pauli operators, $\{I,X,Y,Z\}^{\otimes2}$. For example, the constant $ZZ$ crosstalk is corresponds to vector $\vec v=(0,\cdots,0,1)$.
A unitary quantum channel $\scC_{u}[\text{A}]=\text{U}\text{A} \text{U}^{\dagger}$ can be represented as a unitary matrix $\hat{\text{U}}$ in this basis, with its action represented as matrix multiplication, $\hat{\text{U}}\kket{\text{A}}$, which is also called superoperator representation. 
The problem of DD sequence finding is then translated into a matrix optimization problem.  The goal for the search is to arrive at a zero vector for the targeted crosstalk term (corresponding to zero average Hamiltonian), by optimizing a binary matrix in Pauli basis with each input indicating whether to apply a specific operator.
 
To construct sequences, we usually consider a finite set of pulses, for example, single-qubit $\pi$ or $\pi/2$ pulses, which are native gates on many quantum platforms. We can thus work within an operator subspace that is invariant under conjugation by the pulse set. 
Let us use the Heisenberg interaction, $H = X_1X_2+Y_1Y_2+Z_1Z_2$, as a pedagogical example:

The subspace spanned by $X_1X_2$, $Y_1Y_2$ and $Z_1Z_2$ is invariant under the action of $\pi$ pulses. The representation of $\pi$ pulses will be diagonal matrices. 
In \figref{fig:frame2} (a), the (initial) crosstalk Hamiltonian, represented by $\vec v_0=(1,1,1)$, is the first column. A pulse $\pi_X$ on the second qubit is represented as the 3-by-3 matrix in \figref{fig:frame2} (b) to get the updated vector representation, $v_1$. The concatenation of the vectors, $[v_0;v_1;v_2;\cdots],$ as shown in \figref{fig:frame2} (a) is also called  toggling-frame sequence representation \cite{PhysRevX.10.031002}. 
Suppose we choose a sequence of pulses $\hat{P}_1,\dots,\hat{P}_L$, at times $\Delta t,\dots,L\Delta t$, the time-averaged Hamiltonian would be $\bar{v}=\Delta t\sum_{l=0}^{L-1}v_l$, where $v_{l}=\hat{P}_{l}v_{l-1}$ for $l=1,2,\cdots,L$. If the time window between pulses, $\{\Delta t_l\}$, is not uniform, then $\bar{v}=\sum_{l=0}^{L-1}v_l \Delta t_{l}$. 
The goal of DD for crosstalk suppression is to get to $\bar{v}=\vec{0}$. 
For the Heisenberg Hamiltonian, a DD sequence $(I\pi_X, \pi_X\pi_Y, I\pi_X, \pi_X\pi_Y)$ achieves such a goal, as the sum of each row in \figref{fig:frame2} (a) is zero. Using the superoperator and toggling-frame representation, one can translate the DD sequence design problem into a discretized optimization problem.

In Table ~\ref{tab:syncopation-matrix}, we tabulate the DD sequences composed of $\pi_x$ and $\pi_y$ pulses that eliminate individual or  XX, YY or ZZ couplings (or combinations thereof) in a ``syncopation matrix''. 
Note that such syncopation sequences would work for all Heisenberg-like Hamiltonians, $c_1 XX+c_2 YY+c_3ZZ,~c_1,c_2,c_3\in\mathbb{R}$.
A pulse sequence $P_1P_2\cdots P_n$ applies pulses at time $(t/n,2t/n,\ldots,t)$ while in a CPMG pulse sequence $P_1P_2\cdots P_n$-CPMG, pulses are applied at $t/n,3t/2n,\ldots,(n-1)t/n$, demonstrating syncopation when used in conjunction.
This matrix is not intended to be comprehensive, but rather to provide a reference for the simplest sequences through which some commonly seen coupling terms that can be eliminated. The specified sequence may decouple more terms, but we look only at XX, YY, and ZZ, and there could be DD sequences that achieve the same decoupling but have not been specified.

\begin{figure}[htbp]
    \centering
    \includegraphics[width = 1\linewidth]{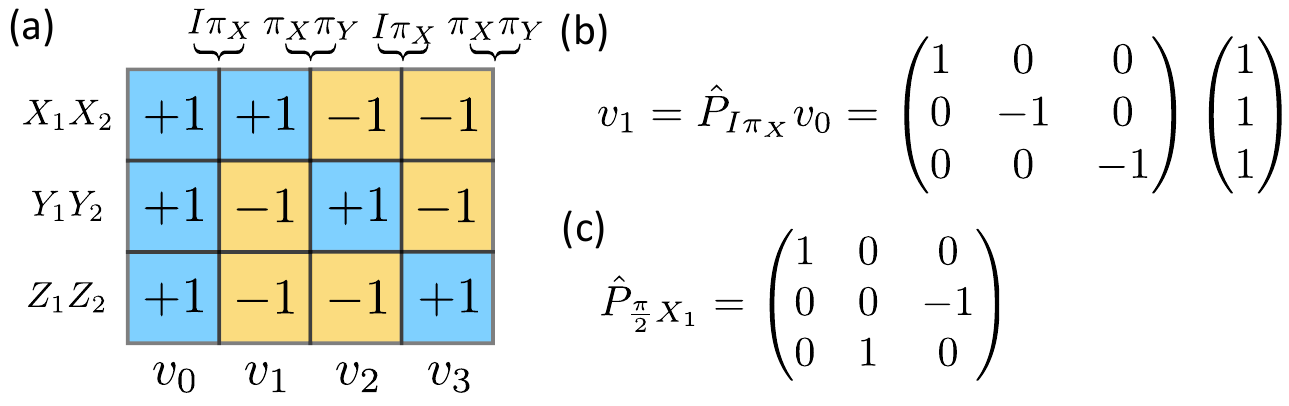}
    \caption{In Pauli basis $\{XX, YY, ZZ\}$, demonstration of DD sequence $(\text{I}\pi_X, \pi_X\pi_Y, \text{I}\pi_X, \pi_X\pi_Y)$ for the Heisenberg Hamiltonian. 
    (a) Toggling-frame sequence representation.  The initial (original) Hamiltonian is represented as the first column.  Each following column shows the updated Hamiltonian after each DD pulse applied in sequential order.
    The zero sum of each row indicates that the corresponding term is averaged out in the effective Hamiltonian, i.e., decoupling achieved.
    (b) Action of the first pulse $I\pi_X$ updated $\vec v_0$ into $\vec v_1$.
    (c) Matrix representation of a $\frac{\pi}{2}$ pulse.
    }
    \label{fig:frame2}
\end{figure}

The Heisenberg-like Hamiltonian is a special case where the three basis operators commute, and hence the DD is exact.  In a general case, the zero-vector goal in an appropriate subspace would still achieve decoupling to the first order in Magnus expansion.  Group structure and features of Pauli operators can be further leveraged to reduce the search space.
Another commonly seen case is when the Hamiltonian also includes single qubit operators. The subspace to consider is beyond the one spanned by XX, YY, and ZZ.  We list a few such Hamiltonians and their working DD sequences as the following:

\begin{itemize}
    \item $H=X_1X_2+Y_1Y_2+Z_1Z_2+Z_1+Z_2:$\\
    DD: $(\pi_x I, I\pi_y, I\pi_y, \pi_y\pi_x, \pi_x I, \pi_y\pi_x, I\pi_y, I\pi_y)$
    \item $H=X_1X_2+Y_1Y_2+Z_1Z_2+Z_1+Z_2+X_1+X_2:$\\
    DD: $(\pi_y \pi_z, \pi_y \pi_x,\pi_y\pi_x,\pi_y\pi_z, \pi_y \pi_x,\pi_y\pi_z , \pi_y\pi_z,\pi_y \pi_x)$
    \item $H=X_1X_2+Y_1Y_2+Z_1Z_2+Z_1+Z_2+Y_1+Y_2+X_1 + X_2:$\\
    DD: $(\pi_y \pi_z,\pi_x\pi_y ,\pi_y \pi_z,\pi_y\pi_z ,\pi_y \pi_z,\pi_x\pi_y ,\pi_y\pi_z, \\
    \pi_y\pi_z ,\pi_y \pi_z,\pi_x\pi_y ,\pi_y\pi_z, \pi_x \pi_y)$\;.
\end{itemize}

\begin{table*}
    \caption{The syncopation matrix identifies which sequences syncopate with each other, for the XX, YY or ZZ static couplings up to length 4.
    X and Y in the sequence refers to $\pi_x$ and $\pi_y$ pulses.
    In a XX sequence pulses are applied at $(t/2, t)$.  In a XX-CPMG sequence pulses are applied at $(t/4,~3t/4)$. 
    In a XXXX-CPMG sequence pulses are applied at $(t/8,~3t/8,~5t/8,~7t/8)$. 
    Each matrix entry indicates the coupling terms that are averaged out.
    }
    \label{tab:syncopation-matrix}
    \begin{adjustbox}{width=\textwidth}
       \begin{tabular}{|r|l|l|l|l|l|l|l|l|l|l|l|l|}
        \textbf{Qubit 1 Sequence} & \textbf{X}X & \textbf{XX-CPMG} & \textbf{XXXX} & \textbf{XXXX-CPMG} & \textbf{XYXY} & \textbf{XYXY-CPMG} & \textbf{YXYX} & \textbf{YXYX-CPMG} & \textbf{YY} & \textbf{YY-CPMG} & \textbf{YYYY} & \textbf{YYYY-CPMG} \\
        \textbf{Qubit 0 Sequence} &  &  &  &  &  &  &  &  &  &  &  &  \\
        \hline
        \textbf{XX} &  & YY, ZZ & YY, ZZ & YY, ZZ & XX, YY, ZZ & XX, ZZ & XX, ZZ & XX, ZZ & XX, YY & XX, YY, ZZ & XX, YY, ZZ & XX, YY, ZZ \\
        \textbf{XX-CPMG} & YY, ZZ &  & YY, ZZ & YY, ZZ & XX, ZZ & XX, ZZ & XX, YY, ZZ & XX, ZZ & XX, YY, ZZ & XX, YY & XX, YY, ZZ & XX, YY, ZZ \\
        \textbf{XXXX} & YY, ZZ & YY, ZZ &  & YY, ZZ & XX, YY & XX, YY, ZZ & XX, YY & XX, YY, ZZ & XX, YY, ZZ & XX, YY, ZZ & XX, YY & XX, YY, ZZ \\
        \textbf{XXXX-CPMG} & YY, ZZ & YY, ZZ & YY, ZZ &  & XX, YY, ZZ & XX, YY & XX, YY, ZZ & XX, YY & XX, YY, ZZ & XX, YY, ZZ & XX, YY, ZZ & XX, YY \\
        \textbf{XYXY} & XX, YY, ZZ & XX, ZZ & XX, YY & XX, YY, ZZ &  & ZZ & XX, YY & ZZ & YY, ZZ & XX, YY, ZZ & XX, YY & XX, YY, ZZ \\
        \textbf{XYXY-CPMG} & XX, ZZ & XX, ZZ & XX, YY, ZZ & XX, YY & ZZ &  & ZZ & XX, YY & YY, ZZ & YY, ZZ & XX, YY, ZZ & XX, YY \\
        \textbf{YXYX} & XX, ZZ & XX, YY, ZZ & XX, YY & XX, YY, ZZ & XX, YY & ZZ &  & ZZ & XX, YY, ZZ & YY, ZZ & XX, YY & XX, YY, ZZ \\
        \textbf{YXYX-CPMG} & XX, ZZ & XX, ZZ & XX, YY, ZZ & XX, YY & ZZ & XX, YY & ZZ &  & YY, ZZ & YY, ZZ & XX, YY, ZZ & XX, YY \\
        \textbf{YY} & XX, YY & XX, YY, ZZ & XX, YY, ZZ & XX, YY, ZZ & YY, ZZ & YY, ZZ & XX, YY, ZZ & YY, ZZ &  & XX, ZZ & XX, ZZ & XX, ZZ \\
        \textbf{YY-CPMG} & XX, YY, ZZ & XX, YY & XX, YY, ZZ & XX, YY, ZZ & XX, YY, ZZ & YY, ZZ & YY, ZZ & YY, ZZ & XX, ZZ &  & XX, ZZ & XX, ZZ \\
        \textbf{YYYY} & XX, YY, ZZ & XX, YY, ZZ & XX, YY & XX, YY, ZZ & XX, YY & XX, YY, ZZ & XX, YY & XX, YY, ZZ & XX, ZZ & XX, ZZ &  & XX, ZZ \\
        \textbf{YYYY-CPMG} & XX, YY, ZZ & XX, YY, ZZ & XX, YY, ZZ & XX, YY & XX, YY, ZZ & XX, YY & XX, YY, ZZ & XX, YY & XX, ZZ & XX, ZZ & XX, ZZ &  \\
        \end{tabular}
    \end{adjustbox}
\end{table*}

The theory accommodates any set of Pauli operators or more complex types of couplings as pulses.  Allowing $\pi/2$ pulses (See Fig.~\ref{fig:frame2} (c) for its representation in the example) can also enrich the plethora of functioning DD sequences potentially of shorter lengths. Furthermore, we have limited ourselves to decoupling sequences that form an identity - they can be inserted into the circuit without additional compiling. If compilation with neighboring gates are taken into consideration, a wider range of sequences is available. 
Taking for example a circuit structure of alternating layers of 1Q and 2Q gates, where the 1Q layers can express any SU(2) rotation. As discussed, a common source of idle time is during nearby 2Q gates, meaning that the idle qubits will have a fully expressive SU(2) rotation on either side of the idle time. This opens up the possibility for non-identity sequences that can target a wider variety of couplings, with potentially fewer pulses. The correction to the identity needs to only be compiled into the following 1Q layer. 

\subsection{Syncopation: time-shifting, frequency-doubling, and more}
A DD sequence (XX,XX-CPMG) as shown in Table~\ref{tab:syncopation-matrix} and presented in previous studies, is a \emph{shifted} syncopation.  Specifically, the two sequences have the same number of pulses but syncopate in timing.

A sequence (XX,XXXX) achieves the same decoupling effect as (XX, XX-CPMG). 
We dub this type of sequence \emph{frequency-doubling} syncopation, which was first proposed in Ref.~\cite{Paz-Silva2016}.  Here frequency-doubling refers to that one DD sequence has twice as many (or any even integer multiplier of) pulses as the other.

The DD sequence presented in Sec.~\ref{sec:superoperator}, $(I\pi_X, \pi_X\pi_Y, I\pi_X, \pi_X\pi_Y)$ that decouples the Heisenberg-like Hamiltonian corresponds to a XX sequence on qubit 1 and XYXY on qubit 2, i.e., (XX,XYXY). This incorporates the techniques of operator alternation and frequency-doubling syncopation.

\subsection{DD targeting ZZ crosstalk}
We identify a family of {\it syncopated} DD sequences capable of decoupling the $ZZ$ crosstalk that is the focus of the experimental part of this study. We describe in detail two approaches of achieving syncopation. We limit the DD sequence to be composed of only $\pi_x$ and $\pi_y$-pulses (X or Y operators) which are experimentally realizable on our platform.  For a two qubit system, the ``frequency multiplication'' syncopation applies a periodic DD pulse sequence to each qubit, but the number of pulses on one qubit is an even multiplier of the other. This is illustrated in Fig.~\ref{fig:two-qubit-XXXX-XX-cartoon}, where we show the (XXXX, XX) scheme; in a ``shifted" syncopation scheme, (XX-CPMG, XX), two pulses are applied to each qubit in an off-beat fashion, as illustrated in Fig.~\ref{fig:two-qubit-XX-XX-cartoon}. The shifted scheme is optimal in the number of pulses. As shown in Table \ref{tab:syncopation-matrix}, 
a number of syncopation schemes can achieve decoupling of ZZ.
Decoupling of XX or YY terms can be achieved by ``operator-alternation" using sequences such as (XYXY, XXXX), but ZZ crosstalk requires syncopation in pulse timing.
If physical $\pi_z$ pulses are supported, from Table~\ref{tab:syncopation-matrix} we can derive from operator cycling that many synchronized sequences, like (XX, ZZ) would also decouple $ZZ$, however, to also protect single qubits from individual $Z$ terms, syncopation in timing is necessary.

We illustrate how such syncopated DD sequences remove $ZZ$ as well as suppress single qubit dephasing by observing the phase accumulation on each qubit, in comparison with the synchronized DD (see Fig.~\ref{fig:two-qubit-XXXX-XXXX-cartoon}) and neighbor only DD (see Fig.~\ref{fig:two-qubit-XXXX-idle-cartoon}) which only removes part of the noise terms. These DD sequences will also be carried out and compared in our benchmarking experiments and discussed throughout the paper.

\subsection{DD targeting other crosstalks}

Our theory generates DD sequences for arbitrary static crosstalk form.  While the device we carried out our experiments in Sec.~\ref{sec:Results} mainly suffers from ZZ crosstalk 
other quantum computer designs give rise to different types of crosstalk.
For example, 
the tunable coupler architecture adopted for superconducting transmon qubits \cite{sete_floating_2021} allows the interaction strength between qubits to be modulated. However, the effective two-qubit Hamiltonian contains both a ZZ term and an XX+YY term and the terms cannot be minimized at the same operating point. This leads to a residual crosstalk for idling qubits. While the strength of this interaction on current tunable-coupler devices is relatively small, it can become more significant as qubit coherence improves. 
This interaction 
takes the form $XX+YY$ or $XY-YX$ depending on the phase of the rotating frame. 
Our theory found the (XX, XYXY) sequence is appropriate for $XX+YY$ crosstalk and (XX, XXXX) sequence for $XY-YX$, 
and the simplest sequence that targets both crosstalk forms is (XXXX, YXYX).

Another example comes from the growing interest in long-range bus couplers \cite{majer_coupling_2007, stassi_scalable_2020}. Such couplers are particularly interesting in error correction codes where long-distance couplings can be valuable, however like the architecture we study here, these coupler designs are typically static meaning they will generate relatively significant crosstalks. Moreover, proposed designs have relatively dense fabrics with long-range connectivity resulting in more complex crosstalk networks. The design proposed by \cite{majer_coupling_2007} produces an XX interaction. The idling crosstalk of this coupling can be suppressed using the syncopated sequences such as (XX, XYXY) as seen in table \ref{tab:syncopation-matrix}. As quantum error correction circuits naturally have relatively long idle periods during measurement, we expect that syncopated decoupling schemes may be particularly useful here.

Other types of crosstalk may call for modifications or generalization of our technique.
An example is the crosstalk present in ion trap quantum devices. 
The entangling gate commonly used in this architecture is Mølmer-Sørensen (MS) gate. During a MS gate between qubit 1 and qubit 2, $X_1X_2(\theta)$, the leading order crosstalk between qubit 1 and its neighboring idle qubit 3 takes the form $X_1 ( \cos\phi X_3 + \sin \phi Y_3 )$.\cite{fang_crosstalk_2022}
The refocusing scheme proposed in Ref.\cite{ParradoRodriguez2021crosstalk} is a spin echo on the idle qubit which is a special syncopated sequence ($IZ$).  While a big family of syncopated DD can eliminate this coupling form, the fact that qubit 1 is subject to the two-qubit gate and cannot readily take a pulse hinders the direct application of these DD schemes. Optimization techniques further incorporates the gate will be need.

\begin{figure*}
    \centering
    \begin{subfigure}{0.235\textwidth}
        \centering
        \includegraphics[width=\textwidth]{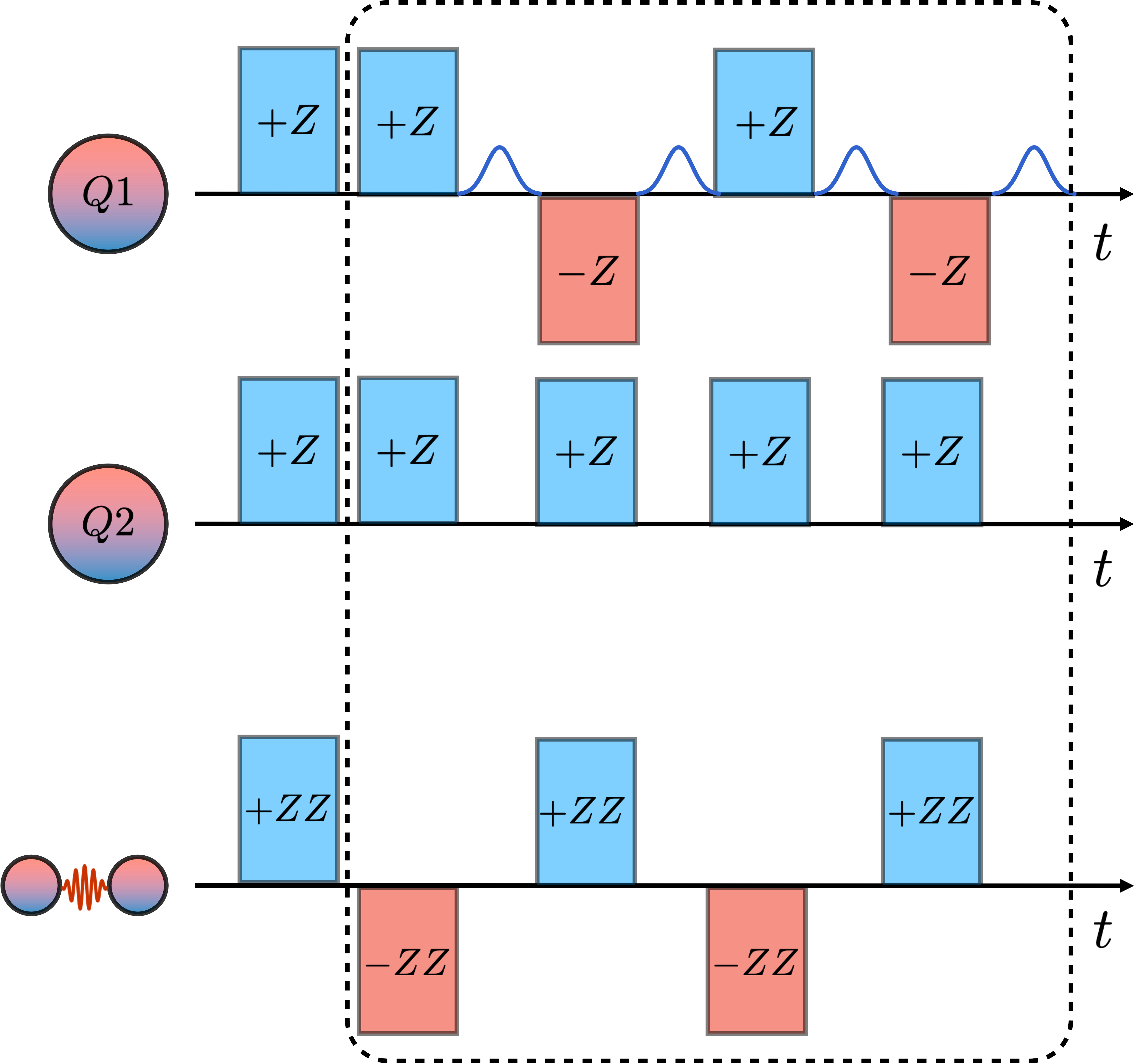}
        \caption{}
        \label{fig:two-qubit-XXXX-idle-cartoon}
    \end{subfigure}
    ~ 
    \begin{subfigure}{0.235\textwidth}
        \centering
        \includegraphics[width=\textwidth]{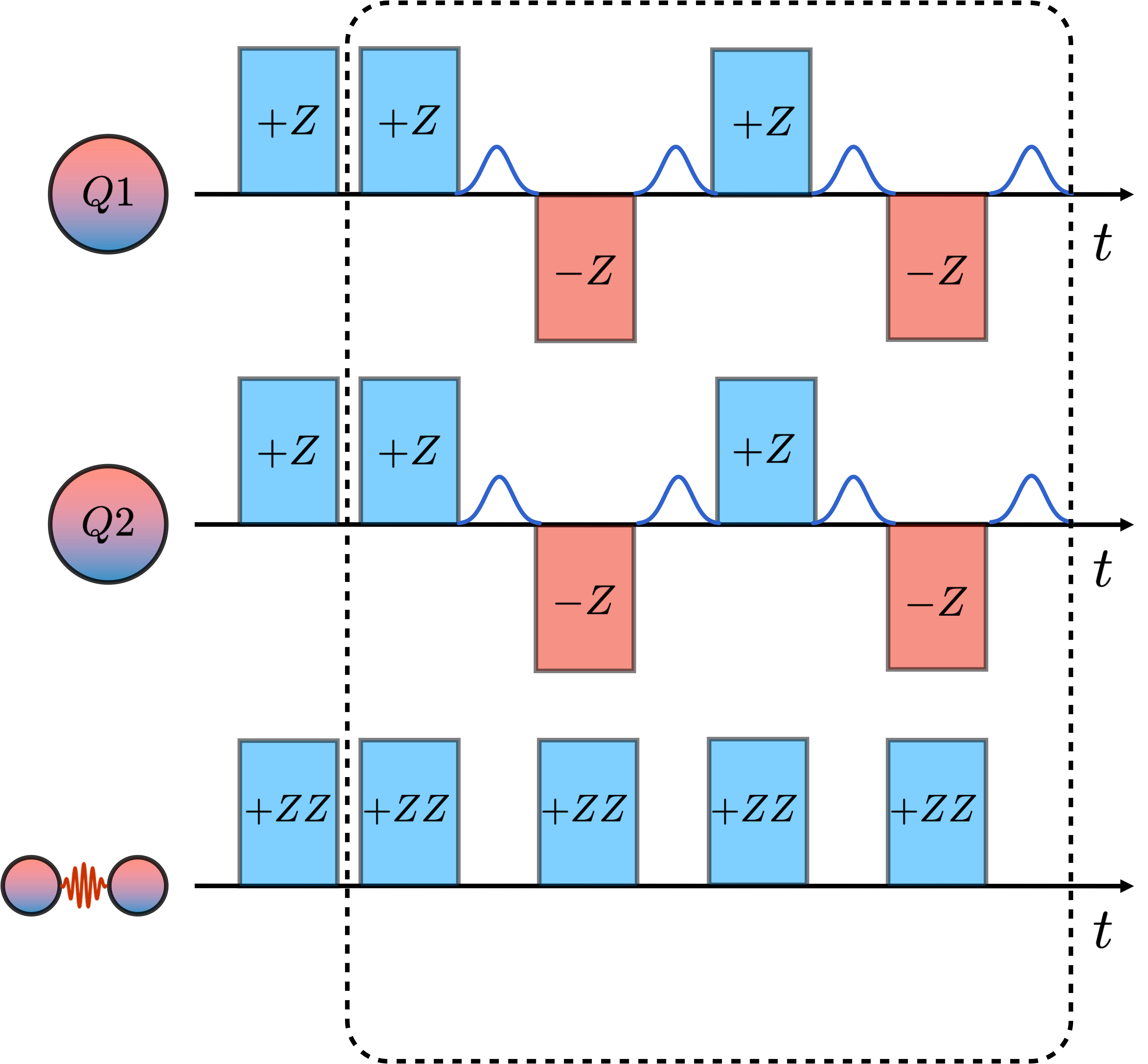}
        \caption{}
        \label{fig:two-qubit-XXXX-XXXX-cartoon}
    \end{subfigure}
    ~
    \begin{subfigure}{0.235\textwidth}
        \centering
        \includegraphics[width=\textwidth]{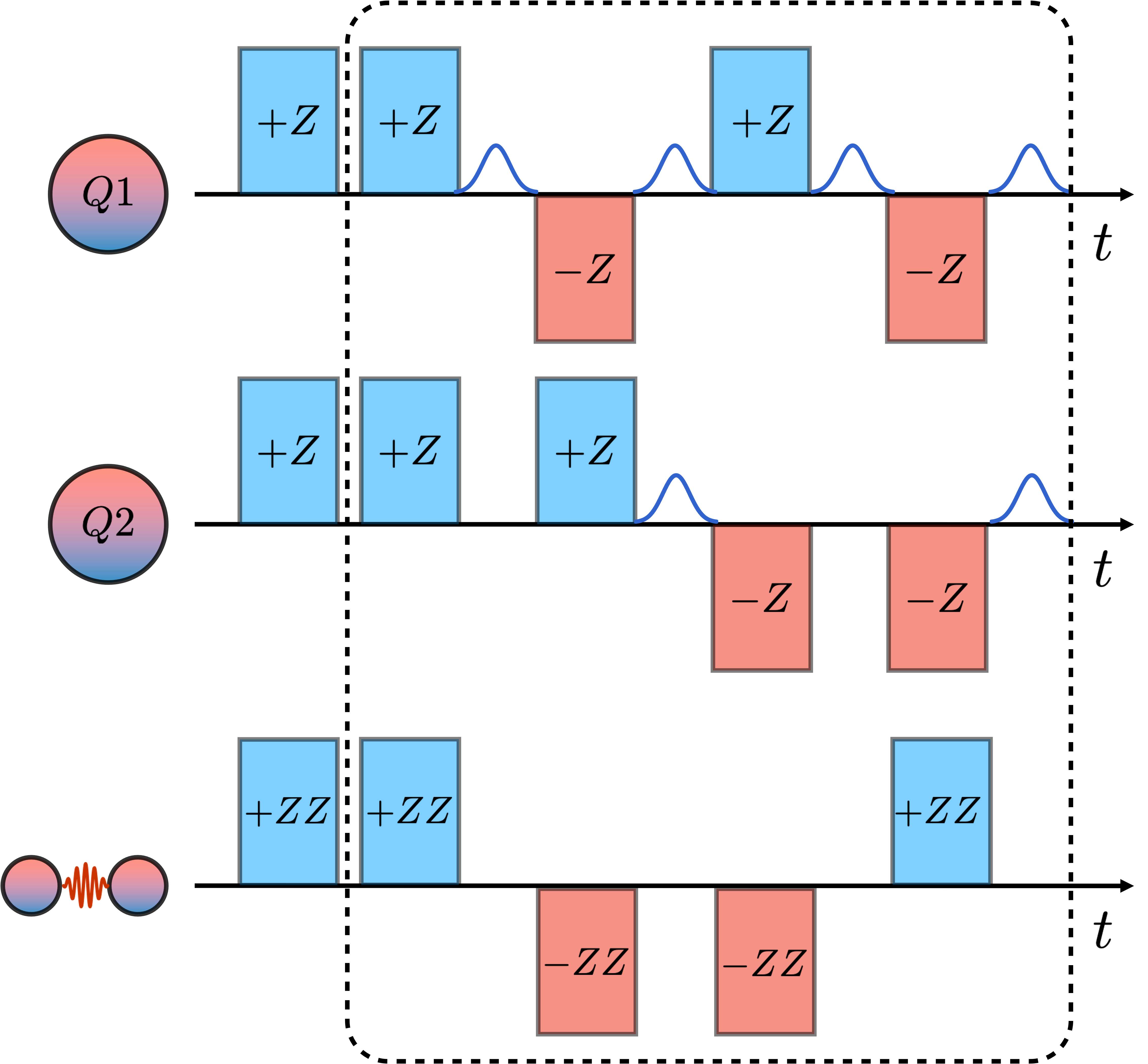}
        \caption{}
        \label{fig:two-qubit-XXXX-XX-cartoon}
    \end{subfigure}
    ~
    \begin{subfigure}{0.235\textwidth}
        \centering
        \includegraphics[width=\textwidth]{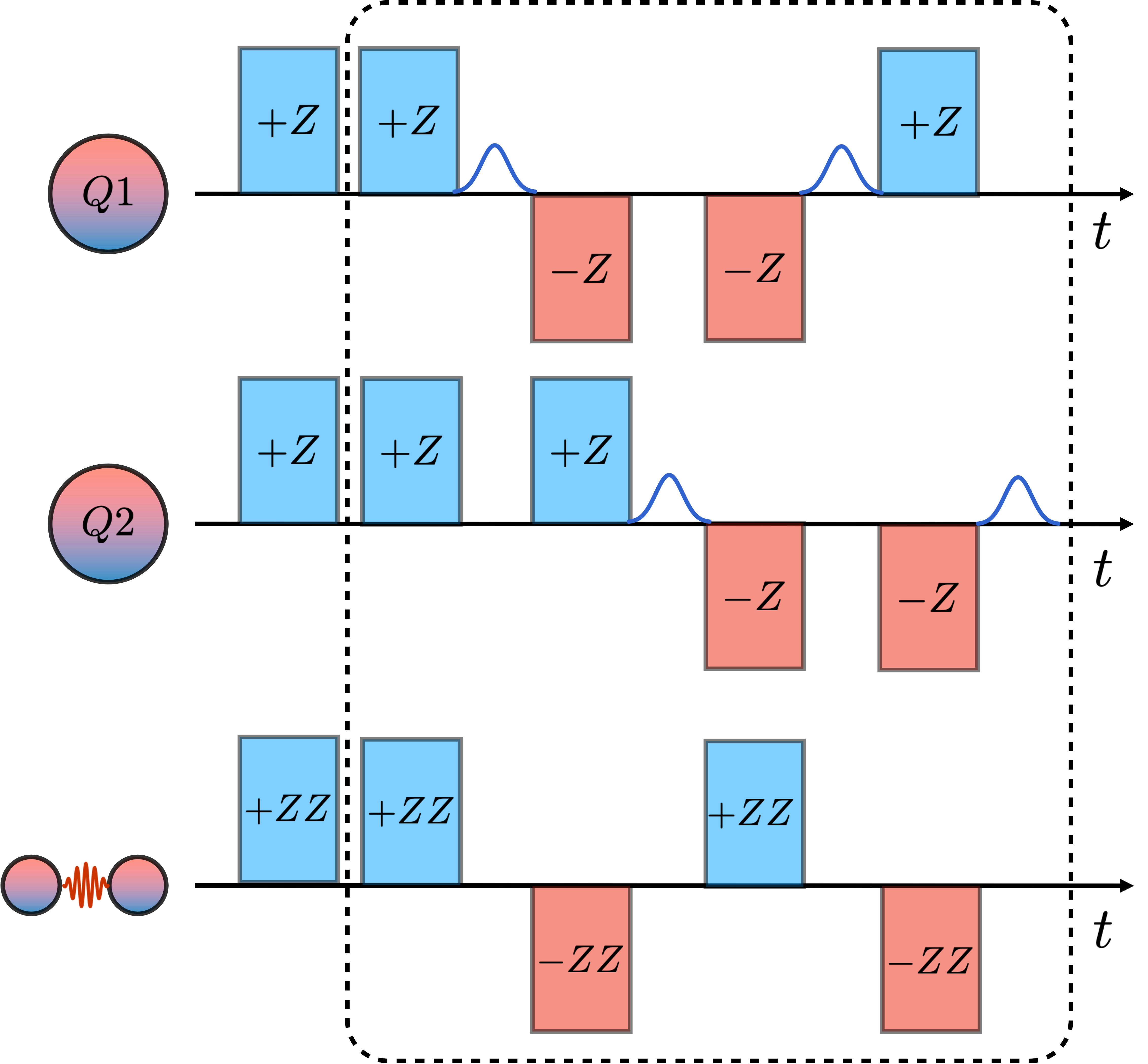}
        \caption{}
        \label{fig:two-qubit-XX-XX-cartoon}
    \end{subfigure}
  \caption{Illustration of DD schemes $(x,y)$ on a two-qubit system subject to a $ZZ$ coupling. DD sequence $x$ and $y$ are applied to each qubit individually. {\bf (a)} Scheme (XXXX,NONE). Decoherence on the first qubit, as well as $ZZ$ coupling, are canceled out by the DD sequence, while decoherence on the second qubit remains. {\bf (b)} Synchronized DD Scheme (XXXX,XXXX). When the same sequence is applied to both qubits synchronously, individual decoherence is removed, but the $ZZ$ coupling between them is unaffected.{\bf (c)} Syncopated DD scheme (XXXX, XX). This scheme averages out $ZZ$ coupling as well as single-qubit decoherence on both qubits. {\bf (d)} Syncopated DD scheme (XX-CPMG, XX). Shifting one sequence can also achieve syncopation, with fewer pulses overall. The pulse duration is exaggerated for illustration purposes.}
  \label{fig:dd_sequences}
\end{figure*}

\section{\label{sec:Results}Results}

\begin{figure*}
    \centering
    \begin{subfigure}{0.25\linewidth}
        \includegraphics[width=\columnwidth]{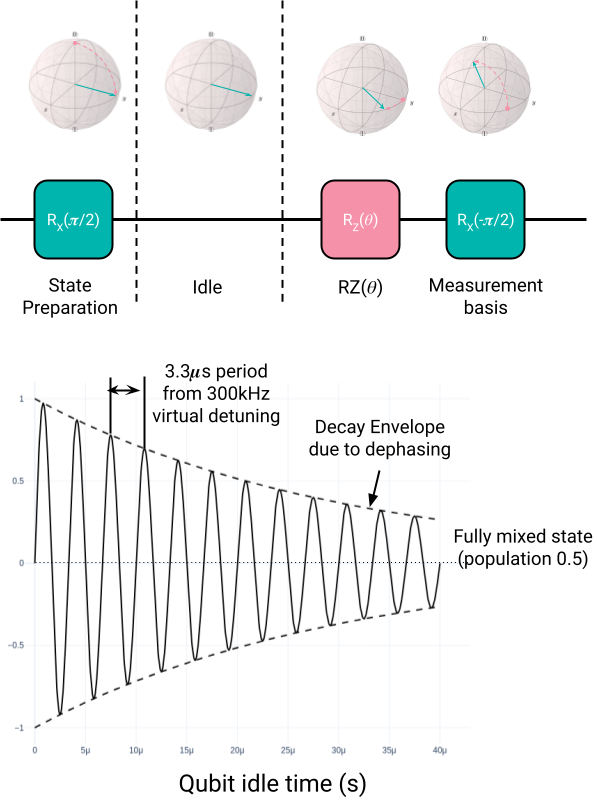}
        \caption{}
        \label{fig:ramsey-circuit-diagram}
    \end{subfigure}%
    \begin{subfigure}{0.375\linewidth}
        \centering
        \includegraphics[width=\textwidth]{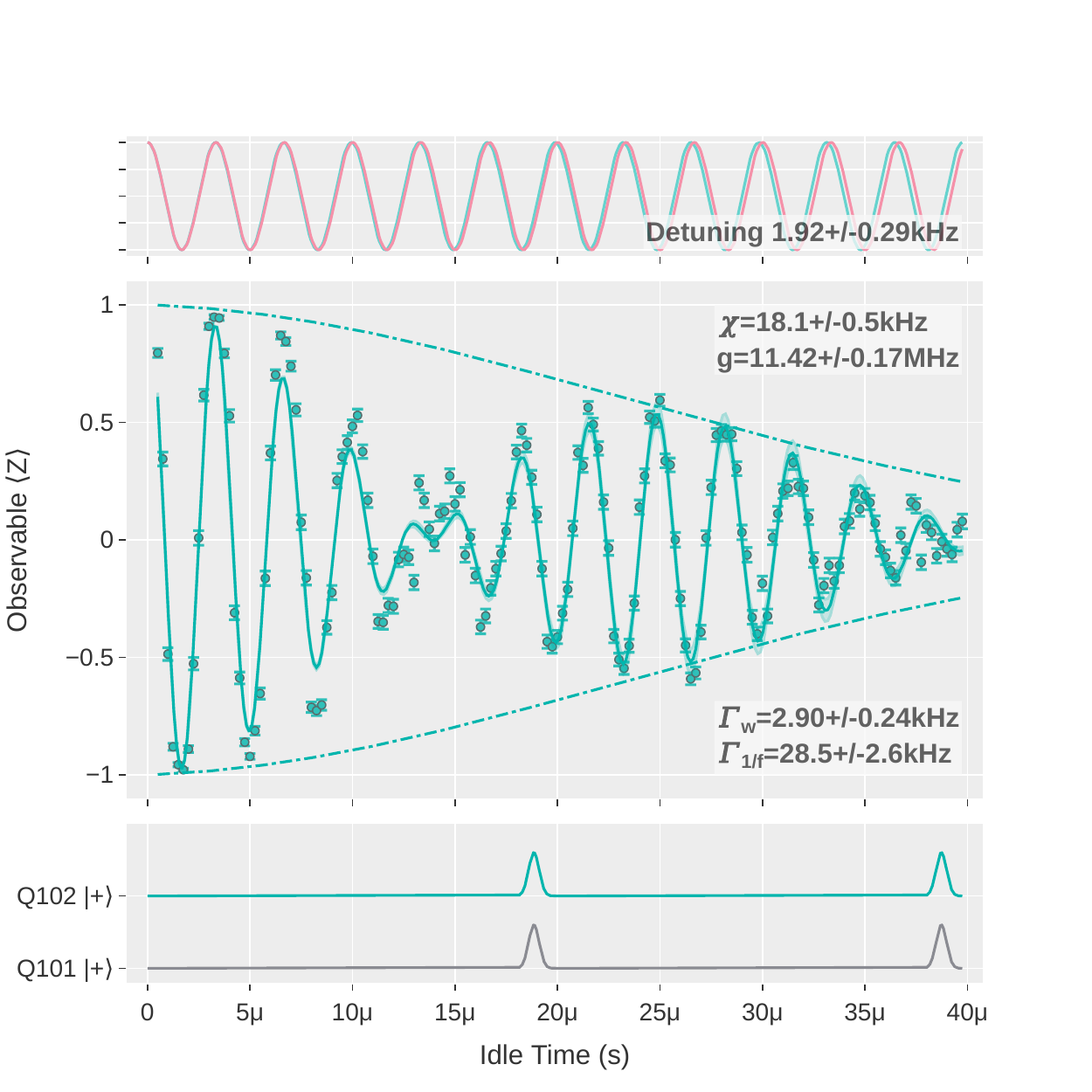}
        \caption{}
        \label{fig:Qubit-102-1-X+_X+-XX-XX-std-std}
    \end{subfigure}%
    \begin{subfigure}{0.375\linewidth}
        \centering
        \includegraphics[width=\textwidth]{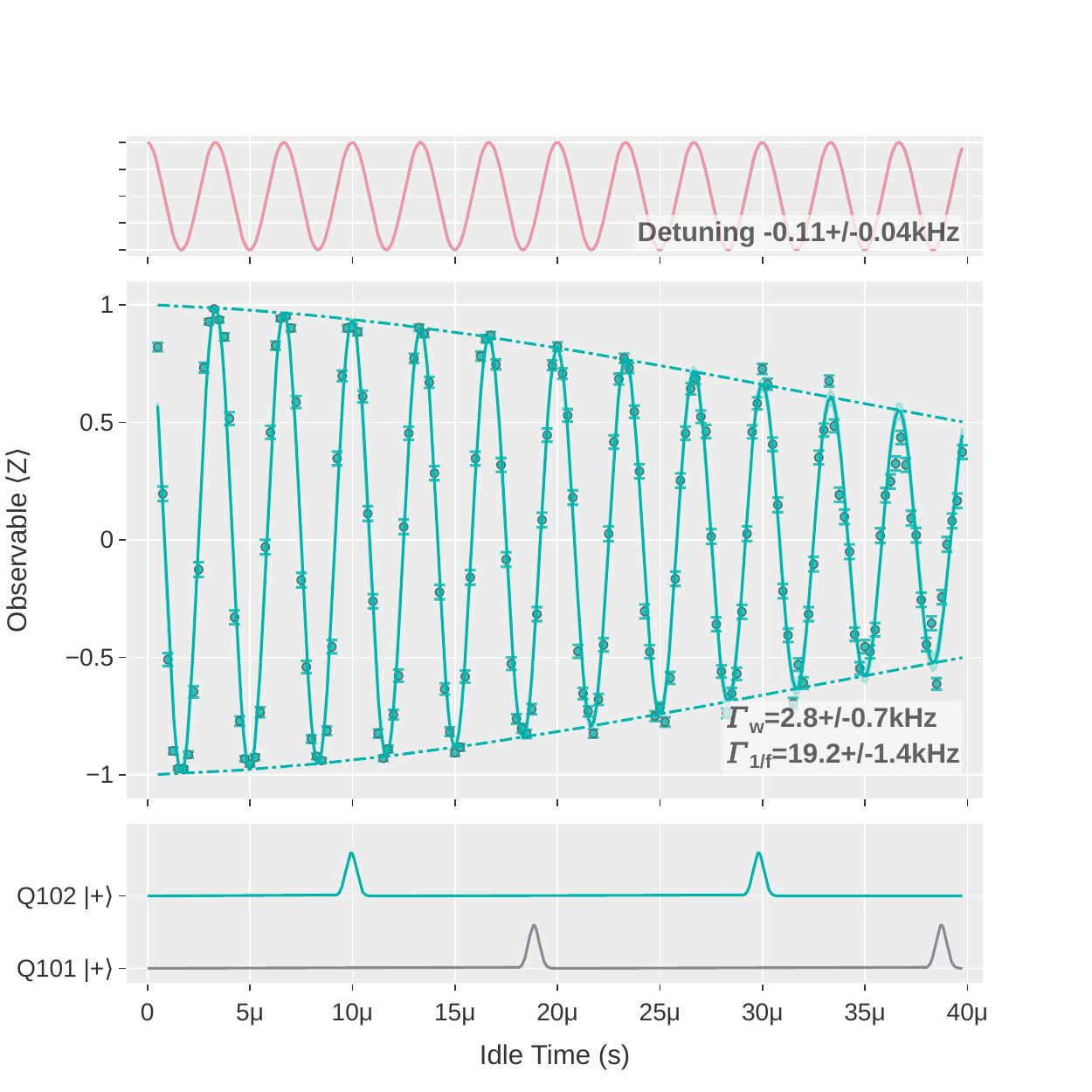}
        \caption{}
        \label{fig:Qubit-102-1-X+_X+-XX-XX-std-cpmg}
    \end{subfigure}
    \caption{The experimental setup is shown in {\bf(a)}. The qubit is prepared in the \Xplus state, remains there for some idle time, in which a decoupling sequence is applied, and undergoes the pre-measurement rotation. The result should be a characteristic Ramsey decay curve, where the physical detuning can be extracted from the frequency of the oscillation and the dephasing rate can be extracted from the envelope. When the neighbour is in the {\bf(b)} \Xplus state, the detuning is negligible but a characteristic beating frequency of 18.1Khz is visible. With the {\bf(c)} syncopated DD sequences, the characteristic beating is suppressed and we recover the expected curve. The $\Gamma_{1/f}$ dephasing rate is reduced from $28.5\pm2.6$kHz to $19.2\pm1.4$kHz, indicating improved protection from decoherence.}
    \label{fig:ramsey-synchronized-vs-syncopated}
\end{figure*}

In the Rigetti Aspen architecture~\cite{Caldwell2018}, two transmon qubits are connected by a fixed capacitive coupler, which leads to an effective $ZZ$ term between the pair \cite{didier_analytical_2018}. The coupling is necessary for 2-qubit gate operation, but remains present while qubits are idling and leads to an unwanted term in the Hamiltonian. Since the source of the coupling is a physical capacitance on the chip, it can be relied on to be static over both the timescale of experiments and the lifetime of the chip. This known source of static crosstalk provides a testbed to apply syncopated dynamical decoupling.

\subsection{Benchmarking: Measuring and mitigating crosstalk in a two-qubit system }\label{sec:benchmarking}
Being able to precisely measure the $ZZ$ coupling magnitude not only helps understand low-level qubit physics and informs hardware design, but also can inform tailoring of error suppression techniques applied to quantum circuits. We begin by applying dynamical decoupling sequences first, to measure the magnitude of a static ZZ coupling and, second, suppress it on the pair of qubits.

\subsubsection{Measuring decoherence}
\label{subsec:detuning}

In order to study the decoherence on individual qubits, we make use of Ramsey $T_2^*$ experiments~\cite{Chiorescu2003}. In these experiments, a qubit is prepared in a superposition state, typically \Xplus or \Yplus, followed by an idle time during which the qubit is subject to noise and (if applied) DD. Finally, a Z-rotation proportional to the idle time is applied, implementing a ``virtual detuning'', and the resulting state is projected onto the measurement axis with the inverse of the preparation operator. The circuit and expected decay is pictured in  Fig.~\ref{fig:ramsey-circuit-diagram}. To obtain a reliable fitting of Eq.~\eqref{eq:fit}, we set a virtual detuning of $\omega_0=300$kHz, which is about an order of magnitude faster than the physical detunings and allows us to resolve dozens of oscillations in our expected dephasing times of 10-100$\mu$s.

The resulting curve without DD can be described by evolution of the expected value

\begin{equation}
\label{eq:fit}
\langle Z (t)\rangle = e^{-\Gamma_{w} t} e^{-\Gamma_{1/f}^2 t^2} \cos{[(\omega_0 + \delta \omega )t]},  
\end{equation}

where $\delta \omega$ is the physical detuning and $\omega_0$ is an introduced virtual detuning. The two decay factors are the empirical decoherence form that prevails on the hardware, an exponential decay due to white noise and stretched exponential decay due to $1/f$-noise, at a decay rate $\Gamma_{w}$ and $\Gamma_{1/f}$, respectively~\cite{Bylander2011}. The observable, $\langle Z (t)\rangle$, is calculated upon measurement in the $Z$ basis. In these experiments, readout mitigation is applied using the method of Ref.~\cite{nachman_unfolding_2020}. 

\subsubsection{\label{subsec:2_qubit_zz}Demonstrating and measuring the effect of ZZ coupling}
We then examine the idle crosstalk between two qubits. The two qubits are initiated in $|+\rangle \otimes |+\rangle$ state, and the observable value of each qubit in a Ramsey experiment is measured. Synchronized DD is applied to the pair, meaning that each qubit experiences precisely the same schedule of decoupling pulses. The decoupling pulses ensure that the crosstalk effect of other nearby qubits is eliminated, but any ZZ coupling between the pair remains, as illustrated in Fig ~\ref{fig:two-qubit-XXXX-XXXX-cartoon}. 

Adding the effect of a $ZZ$ coupling $J$ to the model described by Eq.~\eqref{eq:fit}, we expect the observable to evolve according to
\begin{equation}
\label{eq:fit2}
\langle Z (t) \rangle = e^{-\Gamma_w t} e^{-\Gamma_{1/f}^2 t^2} \cos{[(\omega_0 + \delta \omega )t]} \cos{(\frac{J}{2} t)}\;,  
\end{equation}
where the decay factors similar to those in Eq.~\eqref{eq:fit} account for the single-qubit decoherence under single-qubit DD, and the rest is derived from the quantum evolution of the initial state under the Hamiltonian
\begin{equation}
\label{eq:Hamiltonian}
H = J Z_1Z_2 + (\omega_0+\delta\omega) Z_1\;.
\end{equation}

The $ZZ$ crosstalk between the qubit pair manifests as a beating in the Ramsey measurement, at a frequency proportional to the crosstalk magnitude, $J$. Such beating is observed experimentally, see Fig.~\ref{fig:Qubit-102-1-X+_X+-XX-XX-std-std}. 
Due to this beating, the system undergoes a rapid loss of fidelity far exceeding the simple dephasing rate. Eq.~\eqref{eq:fit2} also provides an efficient way of measuring the crosstalk.  By fitting the data into the equation with fitting parameters $\delta \omega$, $\Gamma_\omega$, $\Gamma_{1/f}$, $\delta\omega$ and $J$, we obtain the two decoherence rates, the residual physical detuning, and the crosstalk magnitude, $J$. A pair of highly detuned transmon qubits coupled by a fixed capacitance gives rise to state-dependent frequency shifts, sometimes referred to as the dispersive shift, $\chi$. The observed beating frequency, $J$, can be related to the dispersive shift and thus the bare qubit-qubit coupling, $g$, using Eq.~\eqref{eq:qubit-qubit-dispersive-shift}, where $f$ is the qubit frequency, $\eta$ is the qubit anharmonicity and the qubits are labeled by subscripts $0$ and $1$. We direct readers to \cite{didier_analytical_2018} for a full treatment. Note that $\chi$ is 2 times the measured beating frequency and is equivalent to $J$ in Eq.~\eqref{eq:fit} and \eqref{eq:Hamiltonian}.

\begin{equation}
\label{eq:qubit-qubit-dispersive-shift}
\chi = \frac{2g^2 (\eta_0 + \eta_1)}{(f_0 - f_1 + \eta_1)(f_0 - f_1 - \eta_1)} \;,
\end{equation}

The bare qubit-qubit coupling, $g$, is directly proportional to the capacitance between the two transmons, $g_C$, and is thus a hardware design parameter which must be carefully controlled. Engineering the strength of $g$ is critical as it determines the speed of two-qubit gates and precise measurements are important for select gate operating points, constructing a physical model of the device and providing feedback for quantum circuit designers.

Using the average of our synchronized DD Ramsey measurements, we extract a beating frequency $\chi = 35.6\pm 1.8$ kHz. This is consistent with the $ZZ$ coupling extracted by measuring the difference in physical detuning caused by preparing the neighboring qubits in $\ket{0}$ and $\ket{1}$ states ($39.1 \pm 0.3$~kHz). This corresponds to a value of $11.34\pm 0.28$~MHz for $g$, which is consistent with the designed capacitance and matches the value of $g$ extracted from the 2-qubit gate operating point as described in Appendix~\ref{sec:appendix:chi} ~\cite{didier_analytical_2018}. The slightly higher measurement uncertainty on the beating frequency is attributed to the larger number of fit parameters in the model, and reduced signal strength over the time frame. We thus note that synchronized decoupling sequences can serve as a viable method for probing the qubit-qubit coupling magnitude on hardware or providing a starting guess for 2-qubit gate bringup. 

Leveraging syncopated decoupling to measure ZZ, or other forms of crosstalk, can be beneficial for devices with large numbers of qubits and complex crosstalks. For example, we may consider the task of measuring ZZ coupling between all pairs of qubits. Using standard JAZZ protocol \cite{takita_experimental_2017, Garbow1982, Ku2020, sagastizabal_variational_2020}, it is possible to measure the ZZ between a pair of qubits using two Ramsey measurements with the neighbour qubit first in the Z+ state, and then in the Z- state. Each measurement yields a detuning and the difference of the detunings is proportional to the ZZ coupling. Using this approach, all the ZZ crosstalks can be measured using $N(N-1)$ Ramsey measurements. By applying syncopation, it is possible to select which crosstalks to decouple, and which to leave in place. This allows for crosstalks to be measured simultaneously, meaning a full crosstalk graph can be characterized with $2(N-1)$ Ramsey measurements as shown in Figure \ref{fig:zz-measurement-schemes}. By using the beating model, another factor of 2 can be gained, reducing the number of measurements fo $N-1$. Thus, leveraging syncopated dynamical decoupling offers a more scalable approach to characterization of static crosstalks.

\begin{figure}[h]
    \centering
    \begin{subfigure}{1.0\linewidth}
        \centering
        \includegraphics[width=\columnwidth]{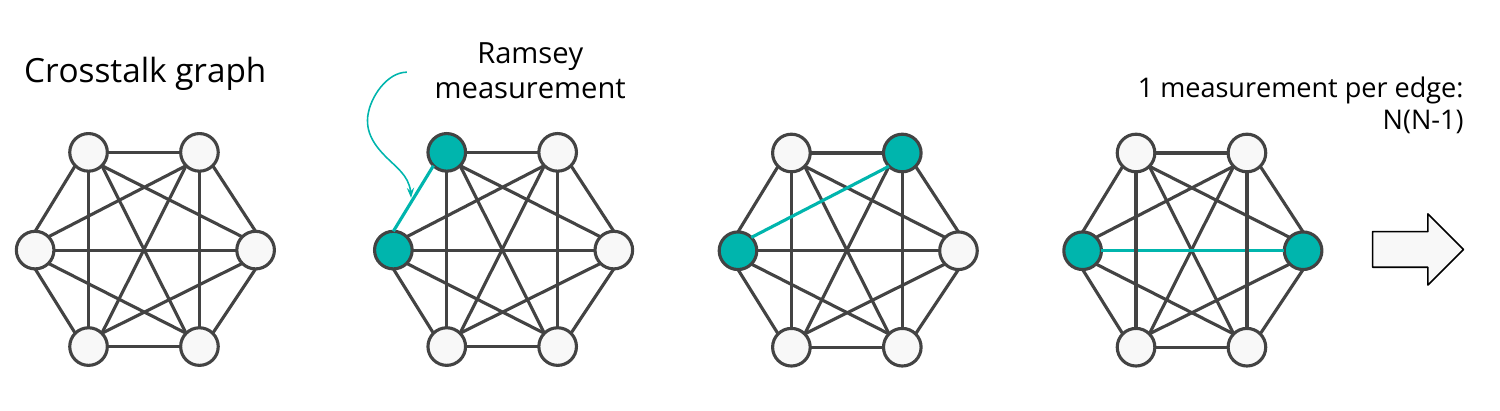}
        \label{fig:zz-one-by-one}
    \end{subfigure}
    \begin{subfigure}{1.0\linewidth}
        \centering
        \includegraphics[width=\columnwidth]{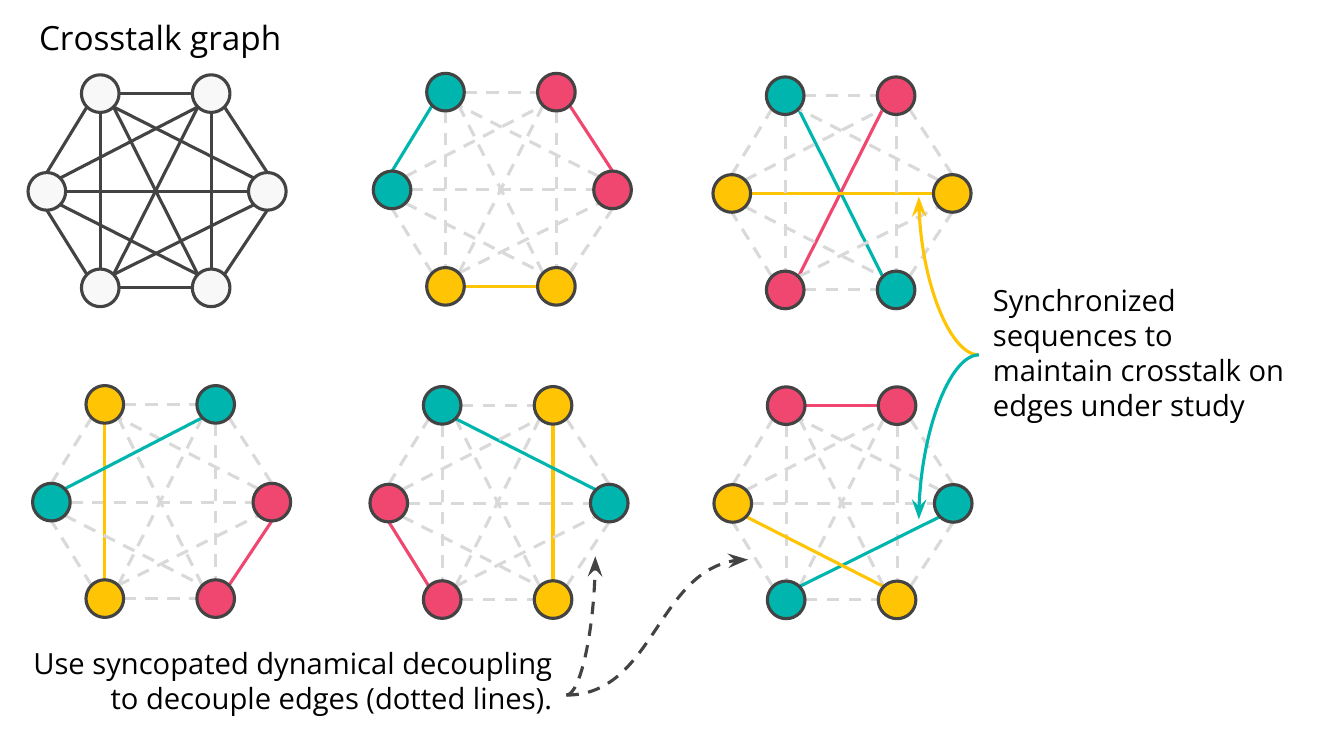}
        \label{fig:zz-syncoapted}
    \end{subfigure}
    \caption{A fully-connected crosstalk graph with 6 qubits is depicted, with crosstalk represented by solid lines. (Top) The JAZZ-based approach to measuring the crosstalk proceeds by changing the state of the qubits one-by-one, and measuring the detuning in each state. $N(N-1)$ measurements (2 for each edge) are required to characterize all crosstalks. (Bottom) Syncopated dynamical decoupling is used to decouple static crosstalks, allowing simultaneous measurement of crosstalks. The decoupled edges are depicted by dotted gray lines, while each pair under study is coloured. This requires $N-1$ measurements.}
    \label{fig:zz-measurement-schemes}
\end{figure}

A related problem is the measurement of qubit frequencies. In the case of a fully-connected ZZ crosstalk graph, the frequency of a qubit will be affected by the state of the entire qubit register, making it difficult to determine the isolated frequency of each. By applying syncopation to decouple the ZZ crosstalk between all qubits, a simultaneous Ramsey measurement can determine the isolated frequency of all qubits using a single Ramsey experiment. Applying syncopated dynamical decoupling in these contexts could yield significant improvements in the number of required experiments. Other forms of crosstalk could be similarly probed, the only modification is to select the appropriate syncopated sequences and measurement basis. Because the device under study suffers only relatively sparse nearest-neighbour ZZ crosstalk, we leave such investigation for future work.

It is a common practice to characterize the dephasing time under the assumption that the $T_2$ decay follows a single exponential. However, if the Ramsey experiments in our study were fit into a single exponential decay, the fitting curve would clearly deviate from the experimental data beyond certain time. Such a feature is consistent through all experiments. The data are fit much better by explicitly including a Gaussian decay factor in the model, to account for $1/f$ noise.  We observe that averaged over all of the experiments, the single exponential decay model fits our data with a reduced $\chi^2$ of 49.8, while the addition of Gaussian decay yields a reduced $\chi^2$ of 38.2, representing improved agreement with the data.  Note that since the reduced $\chi^2$ measure is `per each fitting parameter', the improvement shown is not merely due to extra degrees of freedom introduced in the second decay factor. 
This technique was also used in Ref.~\cite{McCourt_LearningNoise_2023} and better fitting were obtained for experiments therein.
We further noticed that in Table.~\ref{tab:syncopation-summary}, with the two-decay-factor noise model, the single-exponential decay is barely affected while the Gaussian decay is suppressed by DD.
This is consistent with the understanding that the single-exponential decay is often caused by white noise, which cannot be suppressed by echo-based techniques like DD.  In fact, as shown in column ``$\Gamma_w$'' in Table.~\ref{tab:syncopation-summary}, the exponential decay rate extracted from experimental data, bare or with various DD sequences, are highly consistent.  On the contrary, the Gaussian decay rate reflects the efficiency of various DD schemes on suppressing single-qubit decoherence.
Such numerics provide further evidence that this two-decay-factor model is a model that better captures the underlying noise in the hardware.  We thus note that our benchmark experiments also provide an efficient tool to probe and characterize single-qubit noise and reveal its underlying physics.  

\subsubsection{\label{subsec:suppressing_zz}Mitigating the effect of ZZ coupling}

While we have shown that dynamical decoupling can be used to isolate and measure the effect of ZZ coupling, during normal operation of the QPU it is desirable to eliminate this unintended effect. To this end, we introduce the syncopated decoupling scheme (Figs.~\ref{fig:two-qubit-XXXX-XX-cartoon} and ~\ref{fig:two-qubit-XX-XX-cartoon}).

To demonstrate the scheme, we repeat the experiment shown in Fig.~\ref{fig:Qubit-102-1-X+_X+-XX-XX-std-std}, but with syncopated, rather than synchronized, pulses. The results are shown in Fig.~\ref{fig:Qubit-102-1-X+_X+-XX-XX-std-cpmg}. In this case, we achieve syncopation by shifting the relative timing of the DD sequences (as depicted in Fig.~\ref{fig:two-qubit-XX-XX-cartoon}), but using the frequency doubling approach (Fig.~\ref{fig:two-qubit-XXXX-XX-cartoon}) yields similar results (see Appendix~\ref{sec:additional-data}). The suppression of the $ZZ$ crosstalk is evidenced by the elimination of the beating pattern, and the recovery of the expected Ramsey decay envelope. 
Experimental results with a set of different initial states is shown in Appendix~\ref{sec:additional-data}, where we observe that the physical detuning caused by the state of the neighbouring qubit is also eliminated.

\begin{figure}[t]
    \centering
    \begin{minipage}{\columnwidth}
    \includegraphics[width=\columnwidth]{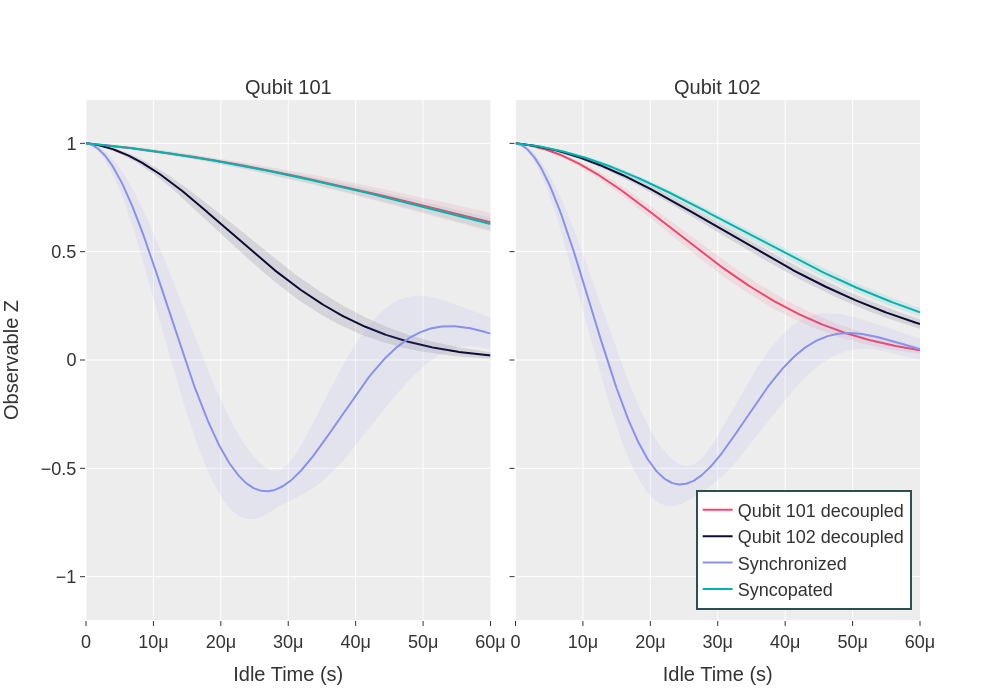}
    {\footnotesize(a)~~~~~~~~~~~~~~~~~~~~~~~~~~~~~~~~~~~(b)\par}
    \end{minipage}
    \caption{The average decay envelope of each qubit when different DD schemes are applied.  DD applied to Qubit 101 alone is shown in red, to Qubit 102 alone is shown in black.  Synchronized and syncopated DD are shown in blue and teal, respectively. The error bands reflect the standard error of the parameter estimate over the set of experiments. The syncopated DD provides the best protection to both qubits.}
    \label{fig:decay-envelopes}
\end{figure}

The experiment is repeated for wide range of DD sequence combinations and qubit states. The full results can be found in Appendix~\ref{sec:additional-data} and are summarized in Table~\ref{tab:syncopation-summary}. We observe that syncopated decoupling suppresses the physical detuning, eliminates the beating caused by the crosstalk, and suppresses the $1/f$ dephasing rate, enhancing the coherent time.

\begin{table*}
    \caption{The fit parameters for the experiments are summarized for syncopated and synchronized decoupling sequences where the neighbour is in state $\ket 0$ or $\ket 1$. We observe that with syncopation, the detuning is eliminated, and the 1/f decay rate is suppressed. The reported parameters are the average of all the experiments performed.}
    \label{tab:syncopation-summary}
    \begin{ruledtabular}
        \begin{tabular}{cccccc}
        \textbf{Qubit} & \textbf{Decoupling} & \textbf{Neighbour} & \textbf{Detuning (kHz)} & $\mathbf{\Gamma_w}$ \textbf{(kHz)} & $\mathbf{\Gamma_{1/f}}$ \textbf{(kHz)} \\
        \hline
        101 & Synchronized & \Zplus & -18.93+/-0.29 & 2.8+/-0.4 & 22.0+/-0.9 \\
        101 & Synchronized & \Zminus & 19.45+/-0.30 & 2.78+/-0.23 & 21.6+/-0.8 \\
        101 & Syncopated & \Zplus & -0.71+/-0.06 & 2.79+/-0.20 & 9.0+/-0.5 \\
        101 & Syncopated & \Zminus & 0.90+/-0.05 & 2.79+/-0.24 & 8.2+/-0.5 \\
        \hline\hline
        102 & Synchronized & \Zplus & -19.75+/-0.30 & 2.8+/-0.4 & 26.0+/-0.8 \\
        102 & Synchronized & \Zminus & 20.02+/-0.30 & 2.9+/-0.5 & 26.1+/-0.8 \\
        102 & Syncopated & \Zplus & -0.103+/-0.035 & 2.79+/-0.21 & 19.2+/-0.5 \\
        102 & Syncopated & \Zminus & 0.403+/-0.033 & 2.79+/-0.23 & 18.8+/-0.5 \\
        \end{tabular}
    \end{ruledtabular}
\end{table*}

The efficacy of the syncopated DD is highlighted by examining the decay envelopes measured for the different decoupling schemes for both qubits, as in Fig.~\ref{fig:decay-envelopes}. While decoupling one qubit only suppresses its dephasing rate and the crosstalk, the neighbouring qubit is still subjected to dephasing. When the same decoupling sequence is applied to both qubits in a synchronized manner, the beating of the decay envelope from the crosstalk indicates the presence of a large coherent error. By syncopating the decoupling sequences, decoherence on both qubits is simultaneously suppressed along with the crosstalk.

\subsection{Application of Syncopated DD to algorithmic circuits}
\label{subsec:application}

Applying dynamical decoupling to improve the performance of application circuits can be non-trivial. A typical approach is to inspect the compiled circuit, locate idling periods of qubits in the circuit and insert the DD sequence of choice~\cite{Niu2022, Qiskit}. However, as demonstrated in Sec.~\ref{sec:benchmarking}, synchronized DD applied to a pair of idle qubits can lead to a perverse outcome where a static coupling remains. To demonstrate the efficacy of our syncopated DD, we constructed a QAOA circuit~\cite{Hadfield2019} that aims to solve the MAXCUT problem on a 4-qubit square as an example. 

To minimize the detrimental effect of two-qubit gate infidelity, we implement a level-one QAOA algorithm, with the optimal angles determined theoretically \cite{Wang_QAOA_Ring}. The algorithm is then composed of state preparation, a layer of phase-separators followed by a layer of single-qubit mixing operators, and a final measurement. We focus on the subroutine of the phase-separation layer, which can be further decomposed into a two qubit native gate on each edge. On the Rigetti Aspen architecture, neighbouring 2-qubit gates are typically not executed simultaneously to minimize operational crosstalk. Rather, they have at least nearest-neighbour separation to ensure high-performance. We note that such idle times are not always readily available in application circuits, however. Tunable-coupler architectures allow for improved operation of neighbouring 2Q gates and thus denser circuits meaning that idle qubits may become less common. The presence of such times is ultimately a feature of the native topology, the algorithm and the operating mode which must be carefully balanced to achieve good performance. Exploring methods of decoupling static crosstalks in circuits with few gaps may be a promising avenue for future work.

\begin{figure}[h]
    \centering
    \includegraphics[width=0.9\columnwidth] {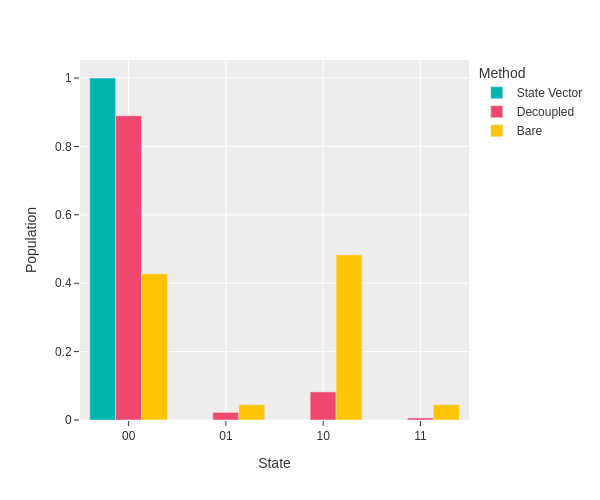}
    \caption{The distribution of bitstrings for spectator qubits during a single cycle of the QAOA circuit.}
    \label{fig:bitstring-distributions}
\end{figure}

In the case of our square lattice, the operating mode naturally leads to a pair of idle qubits in the circuit while the phase-separator is applied to the other two. The phase separator is implemented using the native CPHASE gate of duration $~200$ns, 
The RX($\pi$) gates that compose the $X$ pulses in DD are calibrated to a duration of $40$ns, which means we can fit up to 4 decoupling pulses within the idle time. As entangling gate times are reduced, we expect that the importance of using sequences with minimal number of pulses will grow.

To verify that the decoupling sequences are protecting the idle qubits from errors during neighbouring 2-qubit gate operation, we examine the distribution of bitstrings from a single cycle. While two of the four qubits undergo a CPHASE gate, we observe the output states of the two spectator qubits, which are prepared in the \Xplus state and return back to \Zplus at the end of the circuit. Without any errors, we expect to find both spectators in the $|00\rangle$ states with certainty. However, we find that the spectators are significantly affected by their ZZ crosstalk, and by the operation of the neighbouring CPHASE gate, resulting in significant errors in the distribution. We introduce syncopated decoupling onto these spectators, and observe that the errors are greatly reduced - the bitstring distribution is far closer to the expected one. The distributions are shown in Fig.~\ref{fig:bitstring-distributions}.

\begin{figure}[h]
    \centering
    \begin{subfigure}{1.0\linewidth}
        \centering
        \includegraphics[width=0.9\columnwidth]{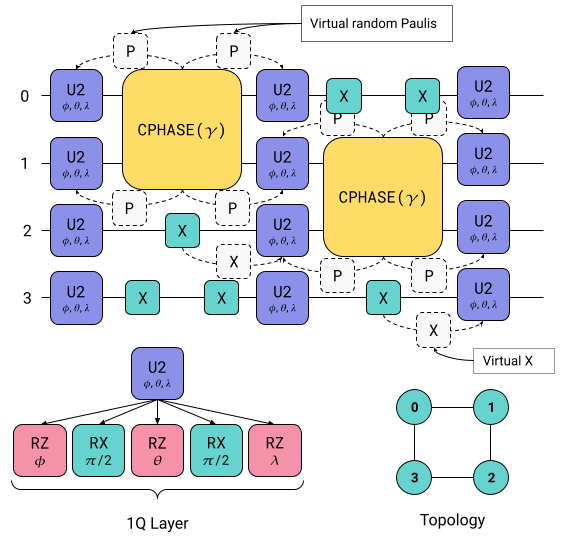}
        \caption{}
        \label{fig:compilation-scheme}
    \end{subfigure}
    \begin{subfigure}{1.0\linewidth}
        \centering
        \includegraphics[width=\columnwidth] {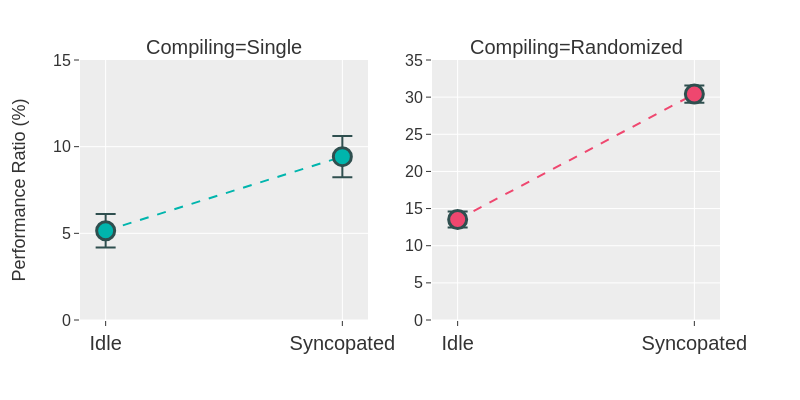}
        \caption{~~~~~~~~~~~~~~~~~~~~~~~~~~~~~~~(c)}
        \label{fig:performance-ratio}
    \end{subfigure}
    \caption{(a) The compilation scheme is shown for a two layers of entangling gates. (b) The performance ratio is shown for the circuit with the qubits left idle and with syncopated dynamical decoupling applied. (c) Syncopated dynamical decoupling is combined with randomized compiling. Applying syncopated DD improves performance with or without randomized compiling, and using the two error mitigation methods in conjunction provides the most significant improvement.}
    \label{fig:qaoa-circuit}
\end{figure}

To further mitigate the effects of coherent errors and state-dependent errors, we additionally introduce the use of random compilation \cite{wallman_noise_2016}. Randomized compiling prevents coherent errors from accumulating, which can improve circuit fidelity. In addition, it ensures that the fidelity improvement is robust under many logically equivalent instances of the circuit. In order to combine syncopated DD with random compilation, the circuit is first structured as layers of 2Q and 1Q gates. Within the cycle, we identify the idle qubits and the static crosstalks between them. In this case, it is clear that the two idle qubits have a ZZ crosstalk, but the crosstalk graph could be more complex for other platforms. The syncopated sequences are inserted on the idle qubits and finally, decoupling pulses which are at the very start or very end of the idle period are absorbed into the 1Q rotation and twirling layer. Our compilation scheme is illustrated in Fig.~\ref{fig:compilation-scheme}.

The results of the experiment are shown in Fig.~\ref{fig:qaoa-circuit}. A common figure of merit in QAOA is the approximation ratio, which the fraction of the true best cost found by the algorithm. However, in this experiment we wish to probe the effect of noise on the QPU, so we define the ``performance ratio'' as the fraction of the noisy divided by the noiseless approximation ratio, which is directly related to the circuit fidelity. Thus, if the QPU were error-free we would report a performance ratio of 1, while a QPU subject to strong depolarizing noise would have a performance ratio of 0. Implementing the bare circuit (i.e. with no error mitigation) results in a low performance ratio, due to readout errors (3\% - 8\%), 1Q and 2Q gate errors (4\% and 10\%), decoherence and crosstalk. Our syncopated dynamical decoupling scheme mitigates the effects of crosstalk errors and decoherence on the idle qubits, and significantly boosts the performance, as shown in the bottom left panel of Fig.~\ref{fig:qaoa-circuit}.  Applying syncopated DD on top of random compilation further doubled the performance ratio, as shown in the bottom right panel of Fig.~\ref{fig:qaoa-circuit}. These two techniques hence not only are compatible with each other, but simultaneously contribute to improving the performance ratio of the circuit.

\section{Concluding Remarks}
Via a discrete search in super operator representation, we identified a family of dynamical decoupling sequences for crosstalk between qubits.
In particular, for ZZ crosstalk that prevails in superconducting transmon-qubit-based architectures, we identified a syncopated decoupling family that can suppress crosstalk between arbitrary pairs of qubits.  We investigated two primitive DD techniques in this family, namely the frequency-doubling DD as well as timing-shifted DD.  In the case of static coupling as in our experiments, they are equivalent in efficiency.  We note that they provide extra flexibility in two scenarios.  One is when the two qubits in frequency-doubling DD are subject to noise of different correlation time, then one can apply the frequency-doubling and assign the sequences with denser pulses to the qubits subject to faster-varying noises.  The other advantage comes in when the idle time in a circuit is constrained (artificially prolonging idle time often comes at the price of introducing more error or lower fidelity), then the mix-and-match of these two DD strategies would provide greater flexibility.

In our two-qubit benchmarking experiments, an (without loss of generality) initial state $\ket{++}$ evolved under nothing but the crosstalk clearly manifested a beating feature, a signature of the presence of the ZZ crosstalk.  We demonstrated how synchronized DD, while not suppressing crosstalk, singles it out from single-qubit noise, and provides a method of precisely measuring the magnitude of the ZZ coupling, which is of great value in calibration.

Through different fitting models, the DD results reveals the two-tone nature of the underlying single-qubit noise.  An exponential decay rate highly consistent through all experiments, with or without DD pins down the white noise that is not responsive to DD and yields more focused observations of the efficiency of different DD schemes in mitigating the 1/f noise.

For a QAOA circuit on a square graph, we demonstrated that syncopated DD greatly improves the performance of the circuit.  DD applied in conjunction with random compilation gives it a further boost.  
In future work it is worth exploring the deeper connection between these two techniques during their interplay.

\begin{acknowledgments}

The work described in Sec.~\ref{sec:benchmarking} was supported by 
the U.S. Department of Energy, Office of Science, National Quantum Information Science Research Centers, Superconducting Quantum Materials and Systems Center (SQMS) under the contract No. DE-AC02-07CH11359 through NASA-DOE SAA 403602.  
The work in Sec.~\ref{subsec:application} was supported by 
the Defense Advanced Research Projects Agency (DARPA) under Agreement No. HR00112090058. 
We are grateful for support from NASA Ames Research Center. Z.G.I. and Z.W. are supported by NASA Academic Mission Services (NAMS), contract number NNA16BD14C.

\end{acknowledgments}

\bibliography{references}
\appendix

\section{\label{sec:detail-on-syncopation}Detail on syncopation and related schemes}

Consider an individual qubit experiencing dephasing (represented by a static $Z$ error). If we apply an even number of $\pi$ pulses about an axis in the $x$-$y$ plane periodically, at the end of the DD sequence all the $Z$ errors can be shown to cancel out. Furthermore, if we consider a neighboring qubit (like in Fig.~\ref{fig:dd_sequences}a), which might be causing $ZZ$ coupling with the original qubit, this $ZZ$ coupling is also averaged out by the XXXX sequence applied to the first qubit. This was experimentally demonstrated in~\cite{Pokharel2018}. Such techniques for suppressing $ZZ$ date back to early NMR literature. For example, in Ref.~\cite{Garbow1982} a bilinear rotation decoupling (BIRD) sequence was proposed where, through pulse manipulation of a spin of a particular variety, an effective $\pi$-pulse could be induced on the subset of spins directly coupled to it, and decoupled the $ZZ$ between the subset and the rest of spins. Recent work in the qubit context implemented a similar protocol where the $\pi$ pulses were achieved with single-qubit controls and demonstrated its effectiveness experimentally.~\cite{Tripathi2022}

However, a drawback of this design is that the neighboring qubit is still subject to single-qubit decoherence.  If the same DD sequence is applied to both qubits in a synchronized manner to mitigate individual decoherence, the $ZZ$ crosstalk will not be averaged out, as shown in Fig.~\ref{fig:two-qubit-XXXX-XXXX-cartoon}.  In a circuit where all qubits are data qubits, it would be more desirable to have a DD sequence that both suppresses decoherence and cancels out the crosstalk between all pairs of qubits.

One way to accomplish this is via ``even multiplier frequency'' DD sequences, in which a periodic DD pulse sequence is applied to each qubit of a pair, but the number of pulses applied to one is an even multiplier of the number applied to the other.  This is illustrated in Fig.~\ref{fig:two-qubit-XXXX-XX-cartoon}, where we show the XXXX, XX sequence; while 4 pulses are applied to the first qubit, only 2 are applied to the second, at half the frequency. 

\section{\label{sec:additional-data}More results on benchmarking: neighboring qubit in \Zplus and \Zminus states}
In Sec.~\ref{sec:benchmarking}, we presented results for two qubit initialized in the $\ket {++}$ state, where the $ZZ$ crosstalk manifested as a beating in the Ramsey oscillation. 
Here we show that our protocol is equally valid for initial states $\ket{+0}$ and $\ket{+1}$, where the physical detuning of a qubit is modulated by the state of its neighbour.
Again, in order to isolate the pair of interest from others, we apply synchronized decoupling pulses to both qubits. The ZZ coupling between them results in the qubit frequency being dependent on the state of the neighbouring qubit, and we thus observe a detuning of $\pm ~20$kHz conditioned on the state of the neighbour, shown in Fig.~\ref{fig:ramsey-synchronized-vs-syncopated-two-pulse}.

\begin{figure*}
    \centering
    \begin{subfigure}{0.33\linewidth}
        \centering
        \includegraphics[width=\textwidth]{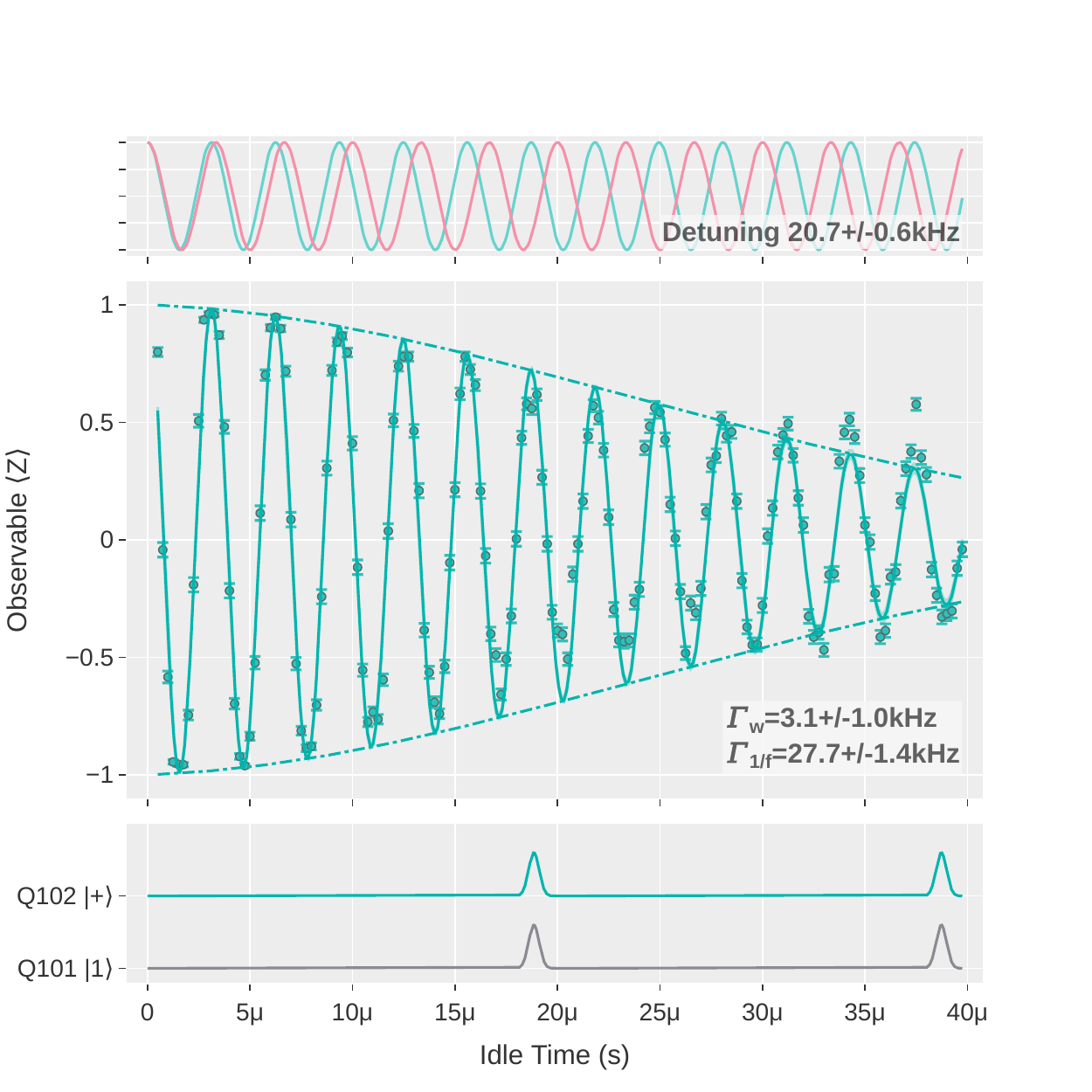}
        \caption{}
        \label{fig:Qubit-102-1-Z-_X+-XX-XX-std-std}
    \end{subfigure}%
    \begin{subfigure}{0.33\linewidth}
        \centering
        \includegraphics[width=\textwidth]{Qubit-102-1-X+_X+-XX-XX-std-std.pdf}
        \caption{}
        \label{fig:Qubit-102-1-X+_X+-XX-XX-std-std}
    \end{subfigure}%
    \begin{subfigure}{0.33\linewidth}
        \centering
        \includegraphics[width=\textwidth]{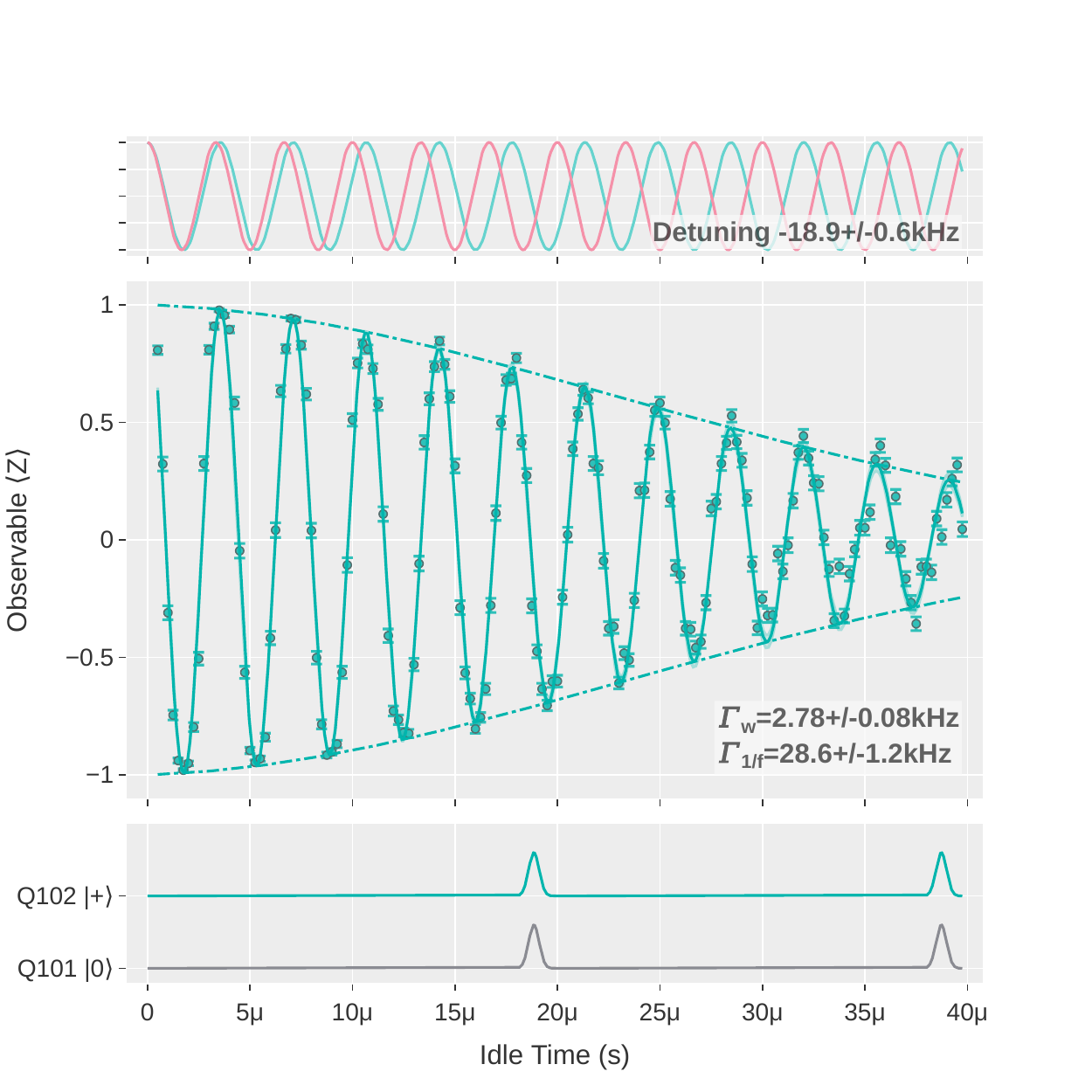}
        \caption{}
        \label{fig:Qubit-102-1-Z+_X+-XX-XX-std-std}
    \end{subfigure}
    \begin{subfigure}{0.33\linewidth}
        \centering
        \includegraphics[width=\textwidth]{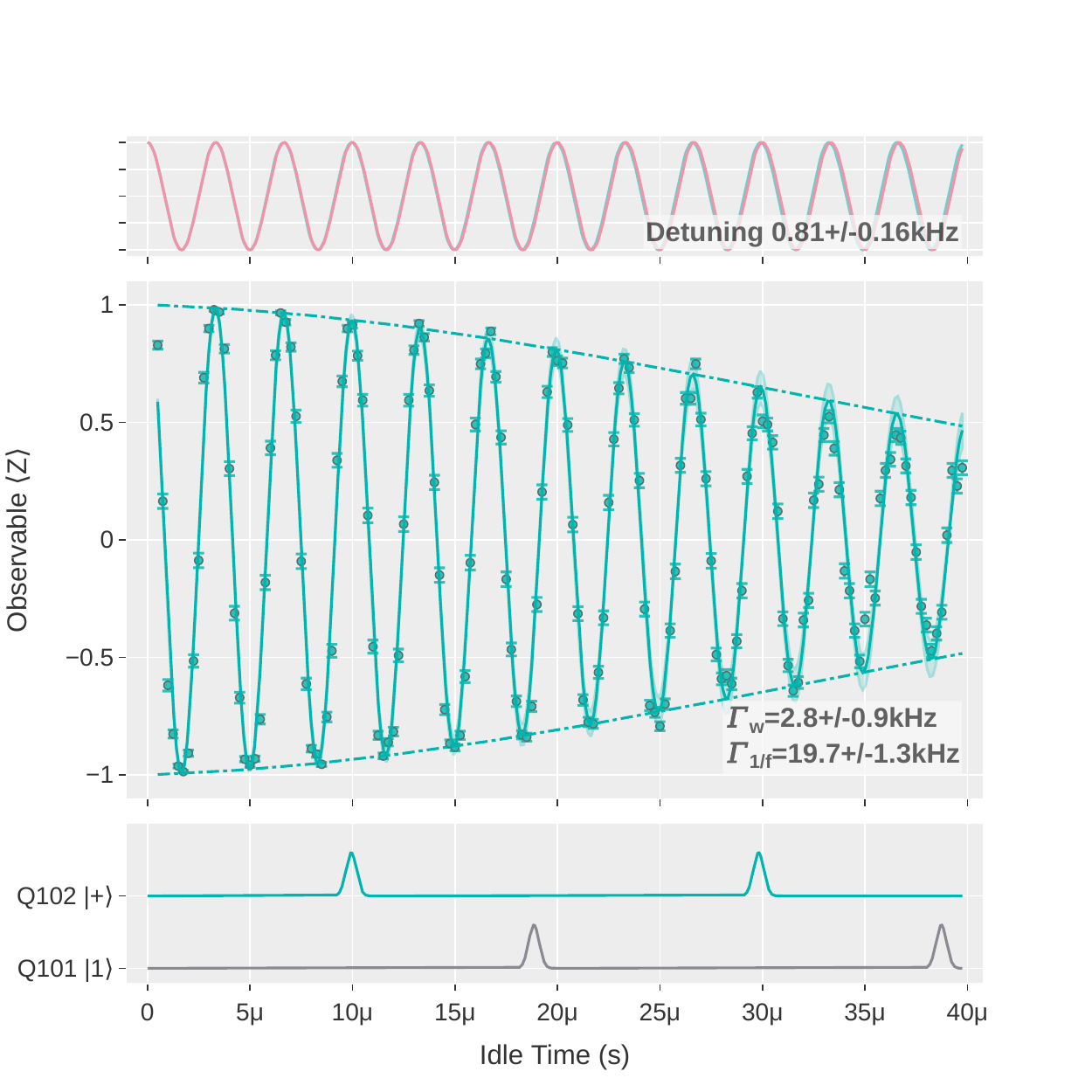}
        \caption{}
        \label{fig:Qubit-102-1-Z-_X+-XX-XX-std-cpmg}
    \end{subfigure}%
    \begin{subfigure}{0.33\linewidth}
        \centering
        \includegraphics[width=\textwidth]{Qubit-102-1-X+_X+-XX-XX-std-cpmg.pdf}
        \caption{}
        \label{fig:Qubit-102-1-X+_X+-XX-XX-std-cpmg}
    \end{subfigure}%
    \begin{subfigure}{0.33\linewidth}
        \centering
        \includegraphics[width=\textwidth]{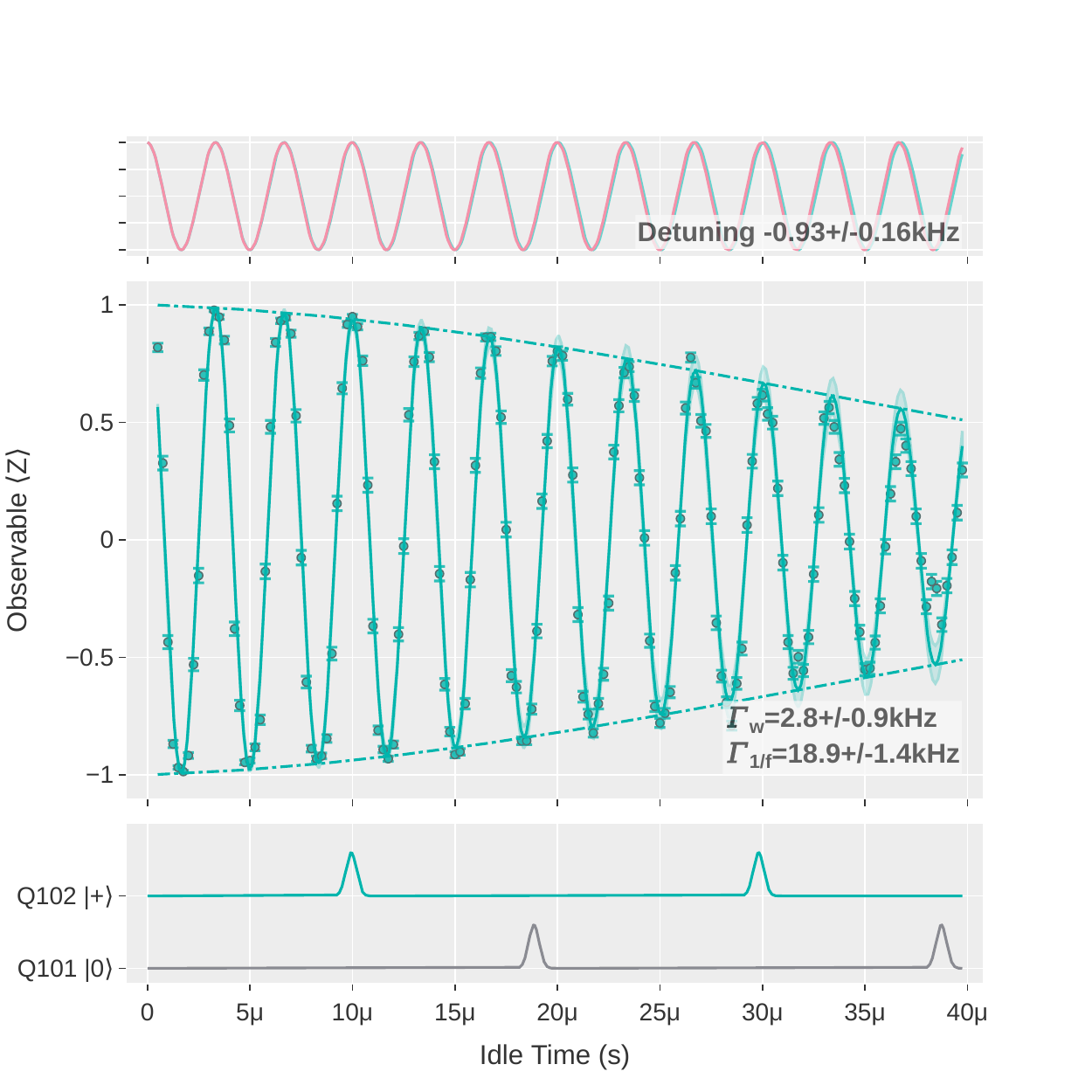}
        \caption{}
        \label{fig:Qubit-102-1-Z+_X+-XX-XX-std-cpmg}
    \end{subfigure}
    \begin{subfigure}{0.33\linewidth}
        \centering
        \includegraphics[width=\textwidth]{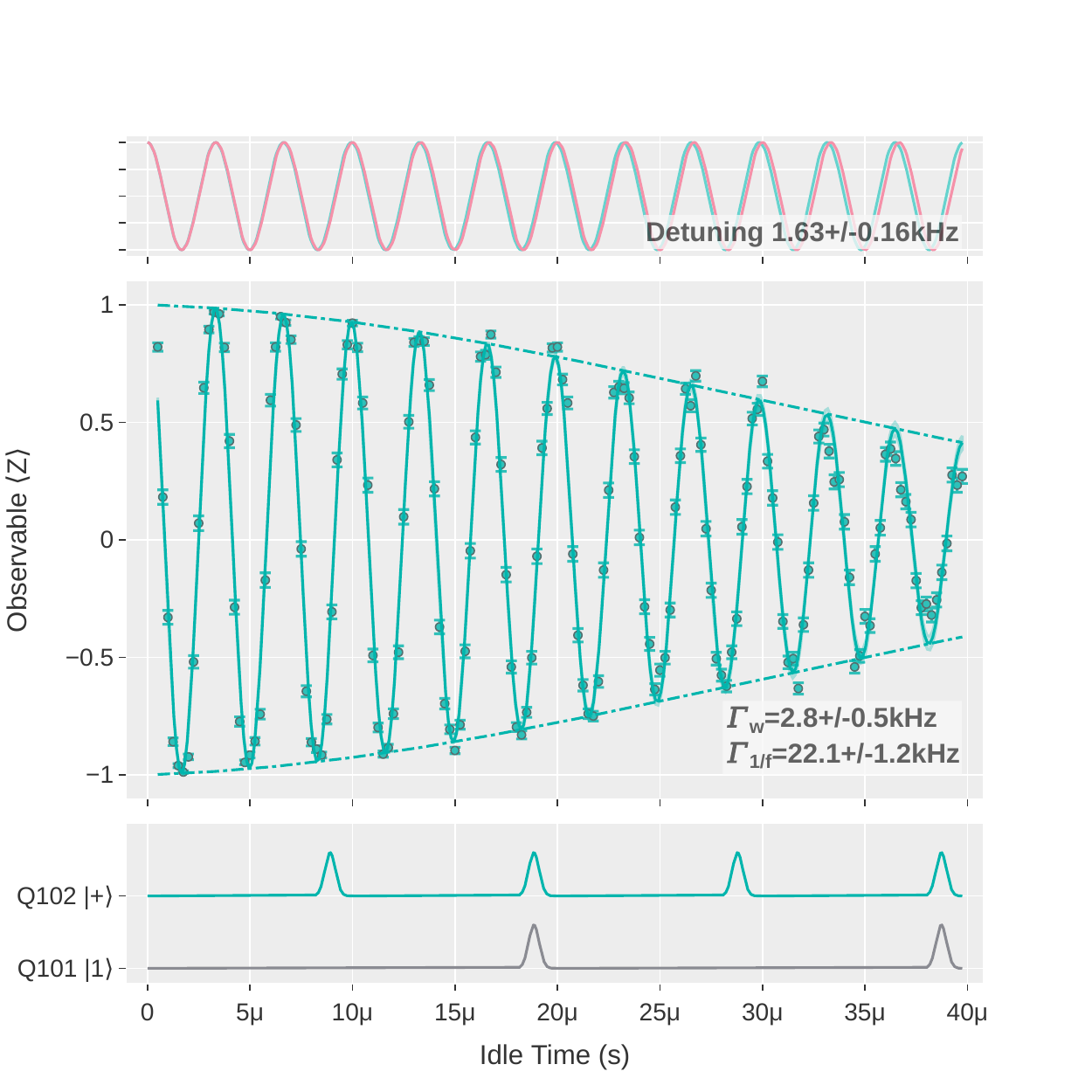}
        \caption{}
        \label{fig:Qubit-102-1-Z-_X+-XX-XXXX-std-std}
    \end{subfigure}%
    \begin{subfigure}{0.33\linewidth}
        \centering
        \includegraphics[width=\textwidth]{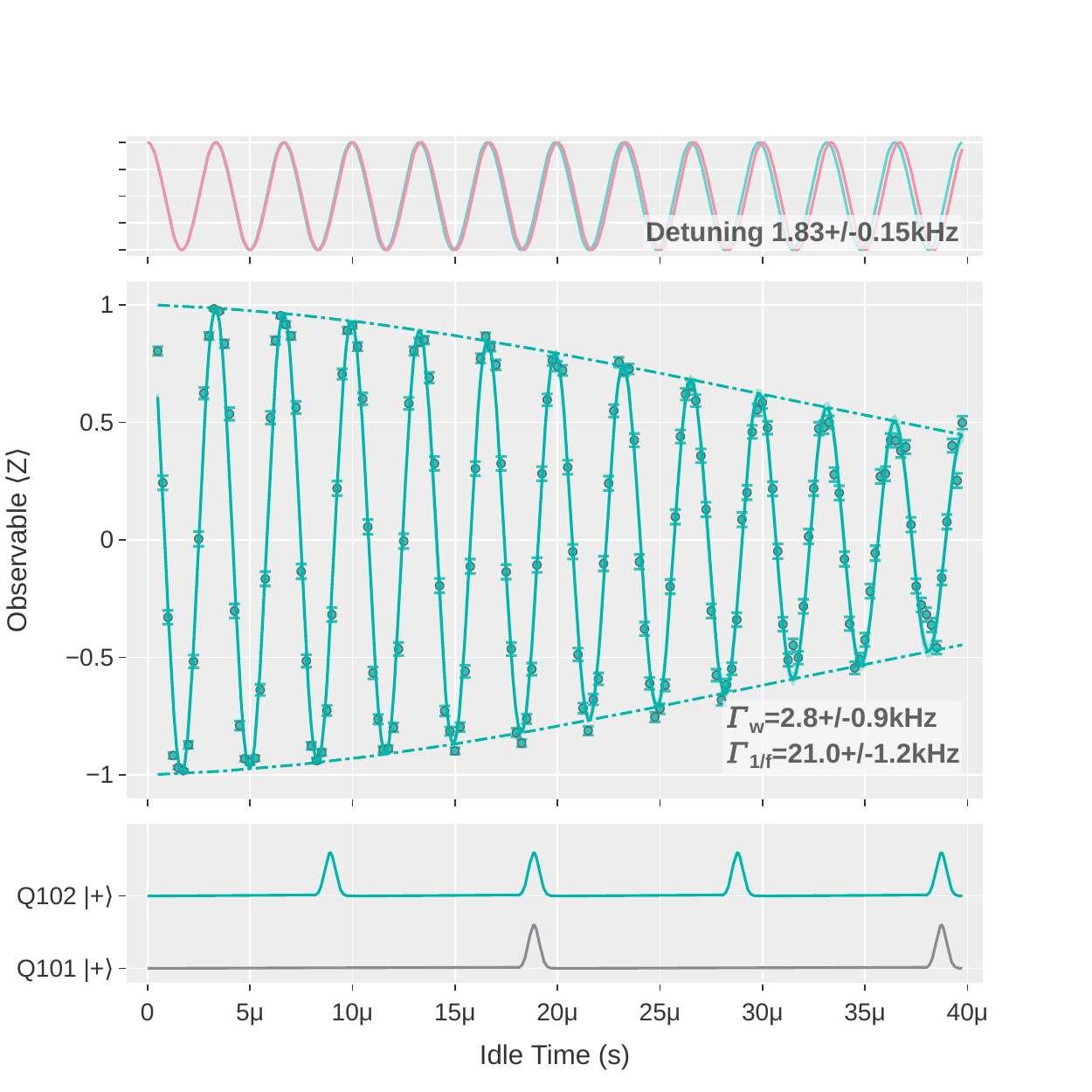}
        \caption{}
        \label{fig:Qubit-102-1-X+_X+-XX-XXXX-std-std}
    \end{subfigure}%
    \begin{subfigure}{0.33\linewidth}
        \centering
        \includegraphics[width=\textwidth]{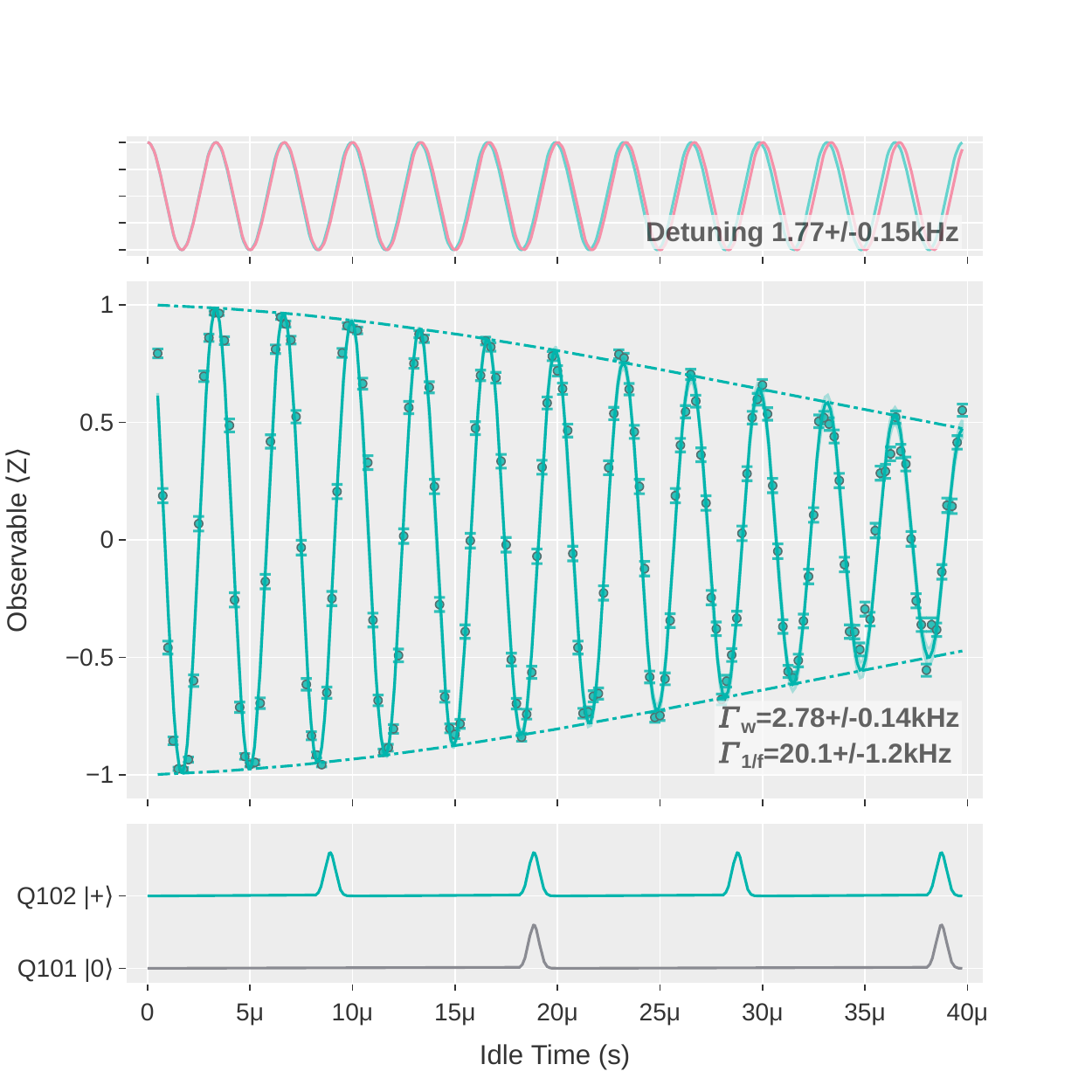}
        \caption{}
        \label{fig:Qubit-102-1-Z+_X+-XX-XXXX-std-std}
    \end{subfigure}
    \caption{(Top row) When neighbouring qubit 101 is in the (a) \Zminus state, qubit 102 experiences a physical detuning of +20.7kHz. When the neighbour is in the (b) \Yplus state, the detuning is negligible but a characteristic beating frequency of 18.1Khz is visible. When the neighbour is in the (c) \Zplus state, qubit 102 experiences a physical detuning of -18.9kHz. The detunings and beating frequency correspond to a $g$ of 11.42 MHz. (Middle row) With the time-shifted syncopated DD sequences, the detuning is eliminated in all cases, and the characteristic beating is suppressed. The decay envelopes are also more gradual, indicating an improved protection from decoherence. (Bottom row) The frequency-doubling syncopated sequences show similar results.}
    \label{fig:all-ramsey}
\end{figure*}

\section{\label{sec:syncopation-on-a-graph}Syncopation on a crosstalk graph}

Qubits are typically laid out in a lattice topology with nodes representing the physical qubits and edges representing interactions between them.
Interactions make entangling gates possible, but are frequently the source of static crosstalk, as is the case in the Aspen architecture, but crosstalk is not limited to the physical topology and may also occur through other mechanisms. Regardless, given a proposed set of crosstalk relationships forming a graph, we can ask how to select a set of decoupling sequences such that the crosstalk is maximally suppressed. This amounts to solving the graph-coloring problem on the crosstalk graph, where each color is a decoupling sequence which syncopates with every other sequence (color). The number of required sequences is the chromatic number of the graph.

\begin{figure*}
    \centering
    \begin{subfigure}{0.33\linewidth}
        \centering
        \includegraphics[width=\textwidth] {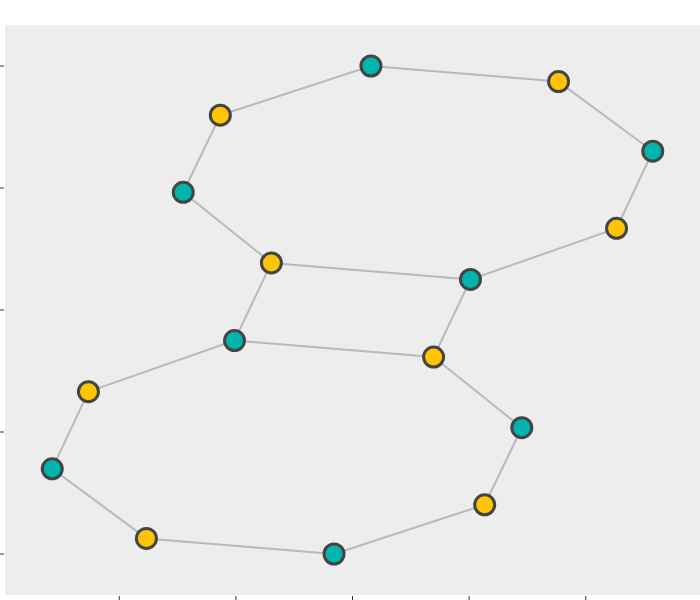}
        \caption{}
        \label{fig:aspen-lattice}
    \end{subfigure}%
    \begin{subfigure}{0.33\linewidth}
        \centering
        \includegraphics[width=\textwidth] {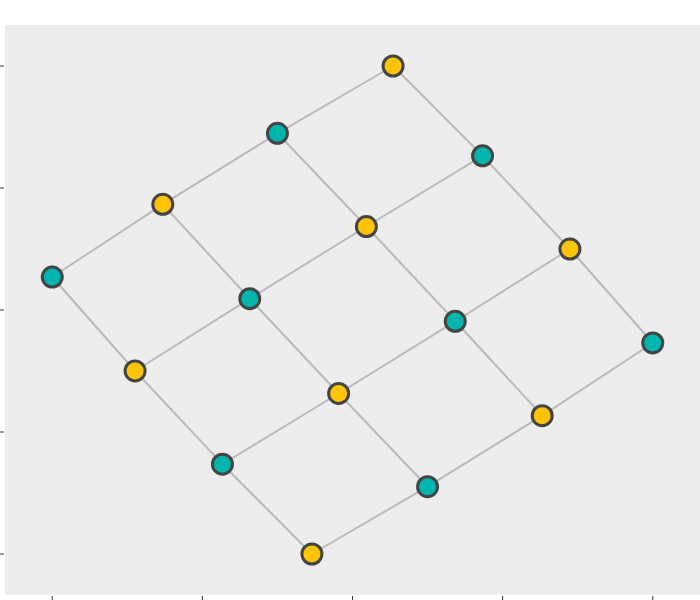}
        \caption{}
        \label{fig:ankaa-lattice}
    \end{subfigure}%
    \begin{subfigure}{0.33\linewidth}
        \centering
        \includegraphics[width=\textwidth] {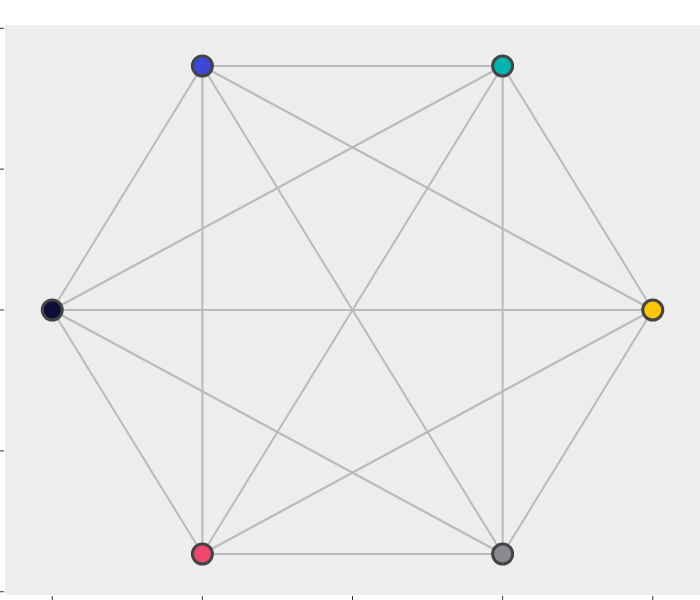}
        \caption{}
        \label{fig:complete-lattice}
    \end{subfigure}%
    \caption{The coloring of the (a) Aspen lattice, (b) square lattice, (c) fully-connected lattice. For the Aspen and square topologies, two syncopated sequences are sufficient to decoupling all the qubits. For the fully-connected model, however, each qubit requires a sequence.}
    \label{fig:ramsey-synchronized-vs-syncopated-two-pulse}
\end{figure*}

While the source and approximate magnitude of the fixed static coupling is known on the Aspen architecture, it could be the case that in other architectures or more advanced designs we have no prior knowledge of static couplings. Thus, we propose that expanding on Section ~\ref{sec:benchmarking}, carefully selected patterns of syncopation could be used to identify unknown crosstalks. Using Table ~\ref{tab:syncopation-matrix}, it is possible to select sequences which eliminate some crosstalks while preserving others, allowing a crosstalk model to be iteratively constructed from the measurement of many decoupling patterns. Hamiltonian engineering of this nature was further explored in refs \cite{tsunoda_efficient_2020} and \cite{bhole_rescaling_2020}.

\begin{theorem}
    Finding the upper bound of minimal dynamical decoupling sequence patterns for qubits with interaction represented by an arbitrary graph is NP-Complete.
\end{theorem}

\begin{proof}
    Let a graph $G(V,E)$ represent the crosstalk between qubits, where each vertex $v\in V$ is a qubit, and an edge $e\in E$ represents crosstalk between them. And label each dynamical decoupling pattern on a single qubit with a different color. Different colors on two connected vertices can cancel the crosstalk. Then, the minimal number of colors needed to color the vertices without any two connected nodes having the same color is the minimal number of dynamical decoupling sequence patterns to cancel all the crosstalks. This number is also called chromatic number $\chi(G)$ of a graph. Since finding the chromatic number of a graph is NP-Complete \cite{DiscreteMath}, that concludes the proof.
\end{proof}

\begin{corollary}
    The minimal number of distinct syncopated dynamical sequences needed to cancel all the $ZZ$-type crosstalk is $2^{\chi(G)}$, where $G$ is the graph representation of the crosstalk, and $\chi(G)$ is the chromatic number of the graph $G$.
\end{corollary}

For example, for a three-colored graph of ZZ crosstalks, the sequences XX, XX-CPMG and XXXX would remove all the ZZ crostalks. Note that all the sequences syncopate with each other.

\section{\label{sec:determining-g}Determining $g$}
\label{sec:appendix:chi}

The combination of Ramsey measurements and decoupling sequences we have described provides a highly accurate way to measure the coupling between a pair of qubits. In order to validate the experimental result from the main text, an alternative approach for obtaining $g$ is described based on the two-qubit gate physics.

Each coupled pair of qubits consists of a tunable-frequency qubit and a fixed-frequency qubit coupled by a transverse coupling, $gXX$. In the weakly-coupled limit in which the qubits are operated, this becomes an effective ZZ coupling. We refer readers to \cite{didier_analytical_2018} for a full treatment, and reproduce the main result below. The Hamiltonian for the interaction picture is well approximated by

\begin{equation}
\begin{split}
    \hat{H}_{int} = &g\sum^\infty_{n=-\infty}J_n \Bigl( \frac{\omega_T}{2\omega_p}\Bigr) e^{i(2n\omega_pt+\beta_n)} \times \Bigl\{ \\
    &e^{-i\Delta t} |10\rangle \langle 01| + \\
    &\sqrt{2} e^{-i(\Delta + \eta_F)t} |20\rangle \langle 11| + \\
    & \sqrt{2} e^{-i(\Delta - \eta_T)t} |11\rangle \langle 02|  \Bigl\}
\end{split}
\end{equation}

Where $J_n$ is the $n^{th}$ Bessel function of the first kind, $\Delta$ is the detuning and $\beta_n$ is the phase.

\begin{equation}
    \Delta = \bar{\omega}_T(\tilde{\Phi}) - \omega_F
\end{equation}

\begin{equation}
    \beta_n = (\tilde{\omega}_T / w\omega_p) sin(2\theta_p) + (2\theta_p + \pi)n
\end{equation}

We can see the Hamiltonian produces three resonance conditions.

\begin{equation}
    2n\omega_p = \Delta(\tilde{\Phi}) \rightarrow |10\rangle \leftrightarrow |01\rangle
\end{equation}

\begin{equation}
    2n\omega_p = \Delta(\tilde{\Phi}) - \eta_T \rightarrow |11\rangle \leftrightarrow |02\rangle
\end{equation}

\begin{equation}
    2n\omega_p = \Delta(\tilde{\Phi}) + \eta_T \rightarrow |11\rangle \leftrightarrow |20\rangle
\end{equation}

Each resonance has an effective coupling strength $g_{eff}^{(n)}$ which determines the Rabi frequency and the resonant linewidth of the interaction at the $n^{th}$ harmonic. This is given by the time-independent prefactor for each term in the Hamiltonian:

\begin{equation}
    g_{eff}^{(n)} = g J_n \Bigl( \frac{\tilde{\omega_T}}{2\omega_p} \Bigr) \leftarrow i\text{SWAP}
\end{equation}

\begin{equation}
    g_{eff}^{(n)} = \sqrt{2} g J_n \Bigl(\frac{\tilde{\omega}_T}{2\omega_p}\Bigr) \leftarrow i\text{CZ}
\end{equation}

A parametric drive that resonantly couples two levels produces swapping in the subspace of those two levels, described by 

\begin{equation}
    \hat{U} = \begin{pmatrix}
    \cos{(\theta/2)} & ie^{-i\phi}\sin{(\theta/2)}\\
    ie^{i\phi}\sin{(\theta/2)} & \cos{(\theta/2)}
    \end{pmatrix}
\end{equation}

where the population exchange, $\theta$ is given by

\begin{equation}
    \theta = 2 \int_0^\tau g_{eff}(t) dt  
\end{equation}

Thus, $g_{eff}$ can be straightforwardly determined from the iSWAP gate time. In the case of qubits 101 and 102, the iSWAP time is 160ns, which corresponds to a $g_{eff}$ of 1.56 MHz. The modulation frequency is 519.23 MHz, corresponding to a AC flux amplitude of 0.602 $\Phi_0$.

The re-normalization constant, $r$ is thus 

\begin{equation}
    r = J_{n}(\frac{\tilde{\omega_T}}{2\omega_p}) = 0.135\;.
\end{equation}


This corresponds to a g of 12.42 MHz, which is in close agreement with the value determined via the dispersive shift (11.3 Mhz).

\end{document}